\documentclass[aps,prx,reprint,groupedaddress,longbibliography]{revtex4-2}

\usepackage{graphicx,xcolor}
\usepackage{amsmath}
\usepackage{amssymb}
\usepackage{amsthm}
\usepackage{physics}
\usepackage{cancel}
\usepackage{booktabs}
\usepackage{overpic}
\usepackage{threeparttable}

\usepackage{xr-hyper}
\usepackage[hidelinks]{hyperref}
\providecommand{\one}{\mathbb{I}}

\providecommand{\coloneqq}{\mathrel{:=}}

\providecommand{\ketbra}[2]{\ket{#1}\bra{#2}}

\DeclareMathOperator{\diag}{diag}

\newtheorem{lemma}{Lemma}
\newtheorem{theorem}{Theorem}

\begin{document}

\title{Pixel-Translation-Equivariant Quantum Convolutional Neural Networks via Fourier Multiplexers}

\author{Dmitry Chirkov}
\thanks{Corresponding author: dmitrii.chirkov@metalab.ifmo.ru}
\author{Igor Lobanov}
\affiliation{School of Physics and Engineering, ITMO University, St. Petersburg, Russia}

\date{\today}

\begin{abstract}
Symmetry is a central source of inductive bias in learning, but in quantum models the symmetry available to a circuit is shaped by the data encoding.
We show that this changes what it means for a quantum circuit to be convolutional.
For address/amplitude image encodings such as FRQI, pixel translations on the finite periodic lattice are represented by pixel cyclic shifts (PCS), which act as modular addition on the spatial index register, whereas many MERA-inspired QCNN templates are equivariant under cyclic permutations of physical qubits.
We formalize this mismatch between PCS and qubit cyclic shifts (QCS), and construct QCNN layers that commute exactly with the PCS symmetry induced by the encoding.
We characterize the resulting layer family as the unitary commutant of the encoded translation operator, yielding a Fourier-space constructive form for PCS-equivariant quantum convolution.
Building on this characterization, we construct a multiscale PCS-QCNN with measurement-conditioned pooling and analyze its trainability with gradient diagnostics that separate parameter-count effects from suppression of the optimization-relevant gradient norm.
Exact PCS equivariance is enforced for the unitary quantum convolutional blocks on their active registers; pooling yields a multiscale/coarse equivariance structure, and in the reported hybrid classifier an unconstrained linear-softmax head maps readout probabilities to labels, so label invariance is learned rather than imposed.
To make the empirical comparison diagnostic of convolutional inductive bias, we use a translated-MNIST benchmark with matched classical CNN/MLP controls that exhibit a separation between translation-aware and dense architectures.
On this benchmark, the PCS-QCNN achieves higher accuracy than a matched non-PCS random-basis quantum control.
Full-MNIST experiments provide a non-translated reference, and finite-shot diagnostics identify a possible train--deploy mismatch at fixed shot budget, making sampling cost a deployment-relevant part of model selection.
\end{abstract}

\maketitle

\section{\label{sec:first}Introduction}

Symmetry is one of the most reliable ways to build inductive bias.
In classical vision, convolutional neural networks (CNNs) hard-code translation equivariance by restricting layer maps to those that commute with pixel shifts (up to boundary conventions), and this architectural constraint is a central source of their sample efficiency and generalization.
Classical literature and textbook treatments support this intuition: translation-aware architectural constraints and weight sharing can improve generalization and sample efficiency relative to fully connected alternatives on vision tasks~\cite{LeCun1989,goodfellow2016deep}.

Quantum machine learning raises the same symmetry question in a setting where the data encoding is part of the model.
Variational quantum circuits can be executed on noisy intermediate-scale quantum (NISQ) devices~\cite{preskill_nisq}, and quantum convolutional neural networks (QCNNs) have been proposed as quantum analogues of CNNs, often inspired by multiscale tensor-network structures such as MERA~\cite{Vidal2006,Cong2019}.
However, in the quantum setting the meaning of translation is not purely a property of the circuit: it is jointly determined by how classical data are encoded into a quantum state and which symmetry the circuit preserves.

Two strands of prior work set the context for this distinction.
First, there is a growing empirical literature on quantum-convolution-style models for classical images, typically using hybrid models with classical preprocessing and relatively small quantum circuits.
MNIST-like benchmarks in this literature use heterogeneous task definitions, resolutions, preprocessing choices, and readout protocols~\cite{Oh2021,Jing2022,Huang2023,Wei2022,Li2020,Li2022,Easom_Mccaldin}; a broader overview of representative settings and reported accuracies is provided in Supplemental Table~\ref{_tab:mnist_lit}.
Second, there is a symmetry-focused literature on equivariant quantum architectures, including QCNN variants designed to be equivariant under cyclic shifts, broader permutation groups, or more general task symmetries with resource-aware implementations~\cite{Cong2019,dasPermutationequivariantQuantumConvolutional2024a,Larocca2022,Meyer2023,Schatzki2024,chinzeiSplittingParallelizingQuantum2024}.
Taken together, these strands expose a conceptual ambiguity: for FRQI and related quantum image representations~\cite{Le2011,Zhang2013}, the spatial symmetry of the encoded state need not coincide with a symmetry of qubit labels.
This ambiguity matters because a circuit can be symmetric in qubit space while failing to respect the translation symmetry of the encoded data.

The analysis focuses on amplitude/address-type image encodings, specifically FRQI-like states (Sec.~\ref{sec:benchmark_encoding}).
On the finite periodic lattice used for the exact symmetry statement, pixel translations are represented by pixel cyclic shifts (PCS), which act as modular addition on the binary-valued index register.
By contrast, many existing QCNN designs enforce commutation with cyclic permutations of physical qubits (or, more generally, subgroups of the qubit permutation group)~\cite{dasPermutationequivariantQuantumConvolutional2024a}.
We refer to this register-level symmetry as the qubit cyclic shift (QCS).
PCS and QCS coincide only for pixel-to-qubit encodings; under address encoding they generally act differently, so a QCS-equivariant circuit need not implement pixel-translation equivariance.
Permutation-equivariant designs remain useful when the task symmetry can be represented as a qubit permutation under the chosen embedding, but for address encodings translation acts natively on the index register and should therefore be enforced directly at that level.
Figure~\ref{fig:general_conv} illustrates this mismatch.
The design principle is that quantum convolution should preserve the translation symmetry present after encoding, not merely a symmetry of qubit labels.
Throughout the paper, ``convolution'' is used in this symmetry/commutant sense rather than in the finite-support-kernel or local-receptive-field sense.

\begin{figure*}
    \centering
    \includegraphics[width= \textwidth]{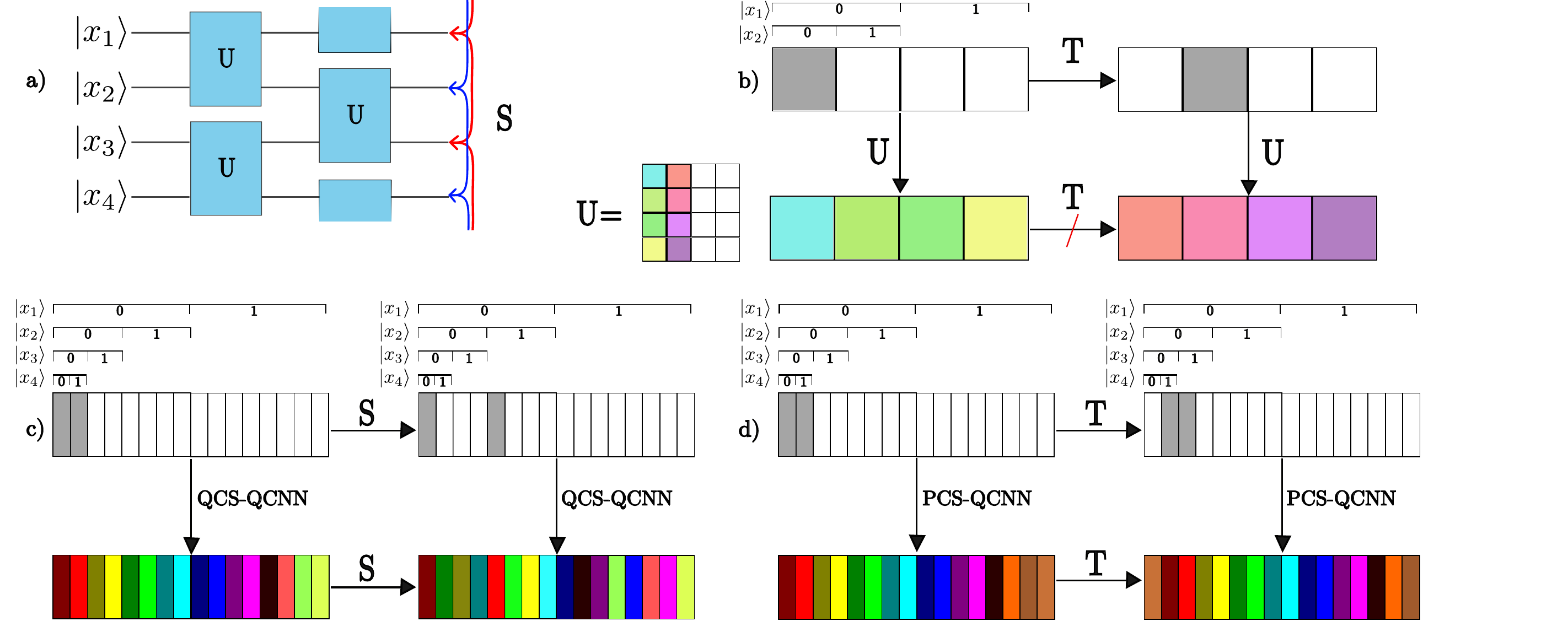}
    \caption{Schematic comparison of qubit- vs pixel-translation symmetry for address/amplitude encodings. (a) A typical QCNN convolution pattern based on repeating the same local unitary on different qubit pairs (adapted from Ref.~\cite{Gong2024}). (b) Under address encoding, translating the input image corresponds to modular addition on the index register, so reusing the same local unitary across qubit positions does not, in general, enforce pixel-translation equivariance. (c) The same pattern can still be equivariant under cyclic permutations of qubits (QCS). (d) Required symmetry: equivariance under cyclic shifts of pixel indices (PCS), i.e., commutation with the translation action induced by the encoding.}
    \label{fig:general_conv}
\end{figure*}

We address this by separating QCS from PCS at the level of the encoded Hilbert space and formalizing when the two symmetries do, and do not, coincide (Sec.~\ref{chapter:arch}).
Once PCS is identified as the relevant data symmetry, the allowed unitary layers are fixed by its commutant.
Since the QFT diagonalizes translation operators, any PCS-equivariant layer can be represented by mapping the spatial index registers to the Fourier basis, applying a block-diagonal transformation that acts independently on Fourier modes (a multiplexer), and returning to the computational basis with the IQFT.
The Fourier diagonalization itself is standard; the key point is that, for address encodings, it must be applied to the pixel-shift operator on the index register rather than to a permutation of qubit labels.
This yields a direct recipe for designing quantum convolution layers whose symmetry is enforced by construction.
We then build a deep PCS-QCNN quantum core with measurement-induced pooling, where intermediate measurement outcomes select mode-dependent blocks in subsequent layers.
The benchmark PCS layers use fully general translation-equivariant Fourier multiplexers; locality or finite spatial support would be an additional structural constraint rather than part of the definition studied here.

A separate practical concern in deep variational quantum models is trainability.
For many circuit families, gradient magnitudes can decay rapidly with system size and/or depth (barren plateaus).
This phenomenon is well documented for highly expressive random circuits~\cite{McClean2018}, while QCNN-style locality can mitigate it in certain regimes~\cite{Pesah2021}.
For our multiplexer-based PCS-QCNN, the trainability analysis separates a mathematically elementary accounting effect from the practical optimization question (Sec.~\ref{sec:pcs_bp_theory}).
Because a fully general Fourier multiplexer contains exponentially many mode-wise parameters in the depth-scaling family, the expected mean per-coordinate second moment can be forced down by parameter count alone.
This coordinate-average effect is distinct from exponential suppression of the full gradient norm or of every active coordinate.
The main results therefore report both the empirical-loss quantum-gradient norm and the per-sample RMS quantum-gradient norm.

MNIST is used as the benchmark for encoding-aligned translation equivariance (Sec.~\ref{sec:mnist_motivation}).
A methodological requirement is that any claimed gain from quantum convolution should be visible in a matched classical comparison.
Standard centered MNIST alone is a weak standalone diagnostic of convolutional inductive bias, because dense classical models can perform well under favorable preprocessing.
We therefore use matched classical CNN/MLP controls to identify a translated-MNIST regime in which translation-aware and dense architectures separate, and then use the same regime to probe the effect of enforcing PCS symmetry in the quantum model against a matched non-PCS random-basis control.
Full-MNIST experiments provide a non-translated reference for the architecture and preprocessing choices, while the translated-MNIST configuration in Fig.~\ref{fig:small_ds_clas} provides the main benchmark for the symmetry-aligned control.

The numerical scope is as follows.
The reported PCS-QCNN is an idealized symmetry benchmark: FRQI-like inputs are initialized directly as dense tensors, fully general Fourier multiplexers are applied as explicit block-diagonal statevector maps, and hardware-oriented structured/compressed PCS layers would define a different model class.
Exact PCS equivariance is claimed for the unitary quantum blocks on their active registers; pooling and the unconstrained linear-softmax head do not impose exact one-pixel label invariance.
The quantum comparison uses a matched non-PCS random-basis ablation rather than a QCS-equivariant baseline, and the finite-shot diagnostics use a single direct-$16\times16$ full-MNIST reference run.

\begin{figure*}
    \begin{overpic}[width=0.49\textwidth]{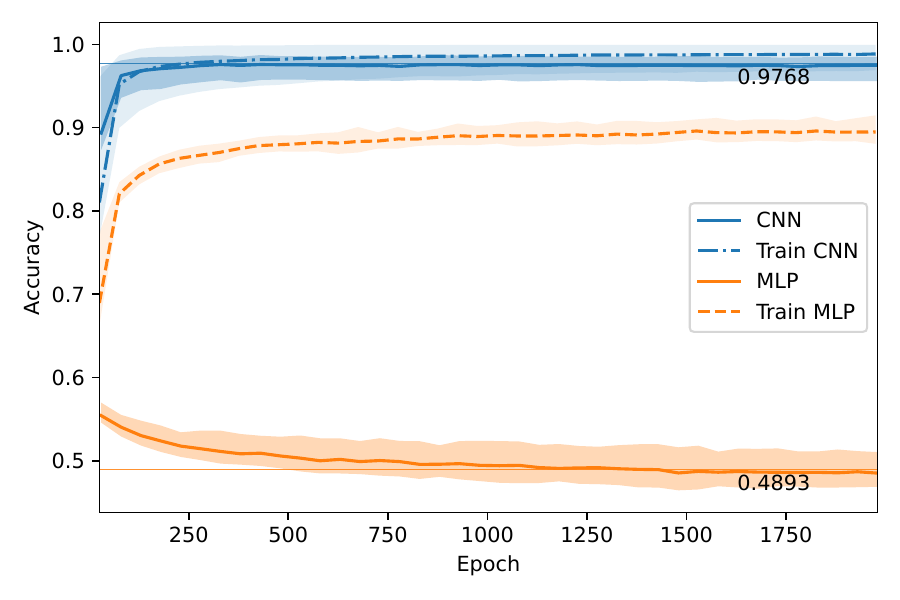}
    \put(3mm,3mm){(a)}
    \end{overpic}
    \hfill
    \begin{overpic}[width=0.49\textwidth]{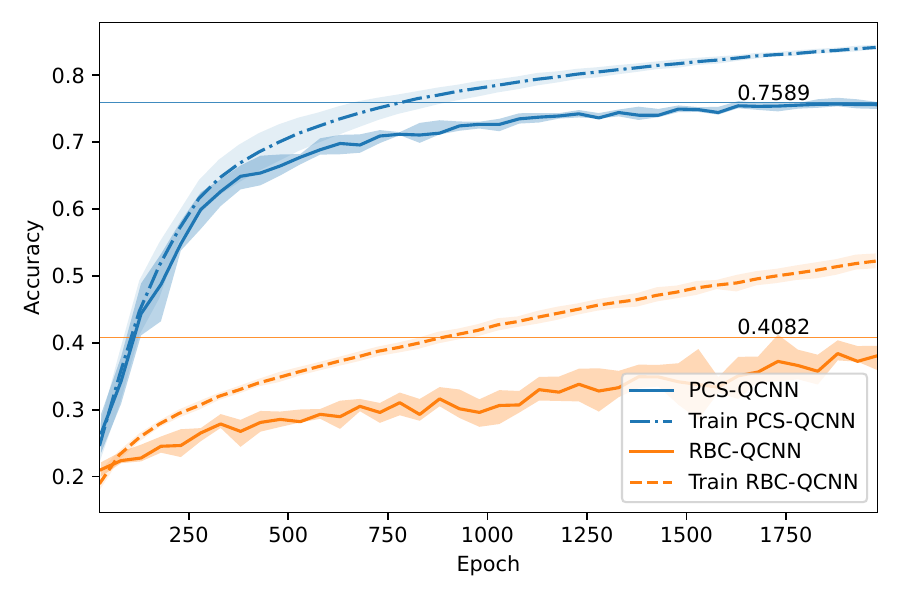}
    \put(3mm,3mm){(b)}
    \end{overpic}
    \caption{Fixed benchmark controls used to expose translation-sensitive inductive bias. Training uses a balanced subset of $1000$ images per class, and the standard MNIST test split membership is kept unchanged. Digits are resized from $28\times28$ to $16\times16$, placed on a $32\times32$ canvas, and randomly translated with maximum offset $8$ pixels per axis. (a) Classical baselines on this translated benchmark: a convolutional CNN reference and a pure-dense MLP control. (b) Quantum models: PCS-QCNN and a matched non-PCS random-basis control (RBC) QCNN on the same translated benchmark. Solid lines show mean test accuracy, train curves show mean train accuracy, and shaded bands show the $25$th--$75$th percentile range over $3$ seeds. The corresponding full-MNIST classical control without translations is shown separately in Supplemental Fig.~\ref{_fig:supp_full_mnist_classical_baselines}.}
    \label{fig:small_ds_clas}
\end{figure*}

Finally, because the PCS-QCNN is evaluated through output probabilities, we test how conclusions drawn from exact statevector inference change under finite-shot sampling.
Even on noiseless hardware, finite-shot sampling introduces stochasticity that can distort loss landscapes and change the effective performance of a trained model.
Accordingly, we report both infinite-shot (exact) simulator accuracy and finite-shot inference.
In the representative finite-shot diagnostic, continued infinite-shot training is associated with a sharper readout-space loss landscape and lower accuracy at some fixed shot budgets (Sec.~\ref{sec:results_shots}).

\section{Symmetry-aligned quantum convolution}\label{chapter:arch}

A symmetry-based definition of convolution gives a common language for the classical and quantum settings.
Throughout the paper we work with discrete images on a periodic lattice.
We use a standard periodic model in which translation is represented by a cyclic shift, so that symmetry claims can be stated without boundary artifacts~\cite{Gray2006,golub2013matrix}.

\subsection{Classical convolution and translation symmetry}\label{sec:classical_convolution}

The classical notion of convolution can be stated purely in symmetry terms: a linear layer is convolutional (in the circular/periodic sense) exactly when it commutes with discrete translations~\cite{Gray2006,golub2013matrix}.
This convention guarantees translation equivariance but does not by itself impose a finite-support spatial kernel.
On a 1D periodic grid of length $N$, let $T_k$ denote the cyclic shift, $(T_k x)_j=x_{j-k}$.
A matrix $A$ is circulant if and only if $T_kA=AT_k$ for all $k$, and in that case $y=Ax$ is a circular convolution with shared weights.
The same perspective extends to multi-channel and multi-dimensional signals: translation-equivariant linear maps are precisely block-circulant operators.

A key structural fact is Fourier diagonalization: translations are diagonal in the discrete Fourier basis, so any translation-equivariant linear map becomes block-diagonal in Fourier space (each Fourier mode transforms independently on the channel space)~\cite{Gray2006,golub2013matrix}.
This is the classical template we mirror below: in the quantum construction, the QFT plays the role of the Fourier basis change and the multiplexer plays the role of per-mode channel mixing.
For a more detailed classical primer (including multi-channel and $2$D formulas), see Supplemental Sec.~\ref{_sec:supp_classical_primer}.

\subsection{Quantum registers, image encodings, and two shift symmetries}\label{sec:qnn_and_shifts}

We fix notation for the two symmetry operations already introduced in Fig.~\ref{fig:general_conv}: the qubit cyclic shift $S$ (QCS) and the pixel cyclic shift $T$ (PCS).
The purpose is to keep the classical and quantum cases parallel: in both settings convolution is defined by commutation with the translation action.
The only difference is that in the quantum setting the translation action depends on the encoding.

We begin by specifying the register layout.
We use an index register (or registers) to represent spatial coordinates and a feature register to store channels.
In the $2$D case we use registers $\mathrm{Reg}\,x$ and $\mathrm{Reg}\,y$ with $n_x$ and $n_y$ qubits, so that $N_x=2^{n_x}$ and $N_y=2^{n_y}$ pixels are represented along each axis.
We use a color register $\mathrm{Reg}\,c$ (grayscale corresponds to one qubit, RGB to three qubits) and an auxiliary feature register $\mathrm{Reg}\,f$.
In the canonical PCS-QCNN analyzed in this paper, the size of $\mathrm{Reg}\,f$ is fixed within a given architecture instance; allowing it to grow across layers is a possible extension discussed in Sec.~\ref{sec:new_features}.

Next we distinguish the relevant encoding families.
The central distinction is whether pixels are mapped to qubits (threshold/pixel-to-qubit encodings) or to basis states of an index register (address or amplitude-based encodings).
The analysis focuses on amplitude-based image encodings, where spatial structure is represented by the binary index.
A convenient representative is an FRQI-like state in which the index registers encode coordinates while a color qubit carries grayscale information:
\begin{equation}\label{eq:frqi_like}
\ket{\psi(x)} =
\frac{1}{\sqrt{N_x N_y}}
\sum_{u=0}^{N_x-1}\sum_{v=0}^{N_y-1}
\ket{u}_x \ket{v}_y \ket{\phi_{u,v}},
\end{equation}
where $\ket{\phi_{u,v}}$ is a single-qubit state whose amplitudes encode pixel brightness.
In the benchmark studied here, each preprocessed pixel is represented by a grayscale value $x_{u,v}\in[0,1]$
(see Supplemental Sec.~\ref{_sec:encoder_preproc} for the exact preprocessing convention).
The encoder maps this value affinely to an angle $p_{u,v}=a+(b-a)x_{u,v}$; the PCS-QCNN experiments use $(a,b)=(0,\pi)$, so
\begin{equation}\label{eq:color_map_impl}
\ket{\phi_{u,v}}=\sin(p_{u,v})\ket{0}+\cos(p_{u,v})\ket{1}.
\end{equation}
Equation~\eqref{eq:frqi_like} is the conceptual normalized form; in the numerical protocol of Sec.~\ref{sec:benchmark_encoding} we use the same local feature states but omit the global $1/\sqrt{N_xN_y}$ prefactor during state initialization and restore the overall normalization when computing the final readout probabilities.
For the symmetry discussion, what matters is that translation of pixels acts as a permutation of the computational basis of the index register.

With this setup, the pixel cyclic shift (PCS) is defined as follows.
In one spatial dimension (a signal of length $N=2^{n}$), the pixel cyclic shift operator $T$ acts on the index basis as modular addition:
\begin{equation}\label{eq:def_T_1d}
T\,\ket{j} = \ket{j+1 \!
\!
\!
\!
\pmod{N}}.
\end{equation}
In the classical case the same action on coordinates is $(T x)_j = x_{j-1}$, consistent with the definition of $T_k$ in Sec.~\ref{sec:classical_convolution}.
In two dimensions we use commuting generators $T_x$ and $T_y$ acting on the two index registers:
\begin{equation}\label{eq:def_T_2d}
T_x \ket{u}_x\ket{v}_y = \ket{u+1 \pmod{N_x}}_x\ket{v}_y,
\end{equation}
\begin{equation}
T_y \ket{u}_x\ket{v}_y = \ket{u}_x\ket{v+1 \pmod{N_y}}_y,
\end{equation}
and similarly in three dimensions one introduces $T_x,T_y,T_z$ acting on three index registers.
These translations act trivially on the feature and color registers.

By contrast, the qubit cyclic shift (QCS) acts on tensor factors: $S$ is defined as a cyclic permutation of tensor factors in a register of $L$ physical qubits:
\begin{equation}\label{eq:def_S}
S\bigl(\ket{q_0}\ket{q_1}\cdots\ket{q_{L-1}}\bigr)
=
\ket{q_{L-1}}\ket{q_0}\cdots\ket{q_{L-2}}.
\end{equation}
Many MERA-inspired QCNN templates can be made equivariant under this symmetry by imposing cyclic weight sharing together with periodic or symmetrized block placements and pooling/readout rules compatible with the same cyclic action.
We will refer to commutation with $S$ as qubit translational symmetry (QCS), and commutation with $T$ (or with $T_x,T_y$) as pixel translational symmetry (PCS).

This gives a precise distinction between layer symmetry and data symmetry.
The operator $T$ represents the data symmetry, namely translation in pixel space, while the operator $S$ represents a register symmetry, namely relabeling of qubits.
A quantum circuit is convolutional for image data only if its symmetry matches the translation action induced by the chosen encoding.
This distinction is the source of the mismatch addressed in this paper.

\begin{lemma}[Mismatch between QCS and PCS]\label{lem:mismatch}
Let an image be encoded either (i) by a pixel-to-qubit map that assigns one pixel value to one qubit, or (ii) by an address/amplitude encoding in which pixel locations are represented by an index register as in \eqref{eq:frqi_like}.
In case (i), the pixel cyclic shift $T$ is implemented by a qubit permutation and can coincide with $S$ (up to a convention for shift direction).
In case (ii), the pixel cyclic shift $T$ acts by modular addition on the binary index and therefore permutes computational basis states of the index register; in general this action is not equal to any fixed cyclic permutation of the physical qubits, and commutation with $S$ does not imply commutation with $T$.
\end{lemma}

Lemma~\ref{lem:mismatch} formalizes the point illustrated in Fig.~\ref{fig:general_conv}(b-d):
for address encodings, enforcing QCS is not the same as enforcing the translation symmetry of pixels.
A complete proof is given in Supplemental Sec.~\ref{_sec:supp_symmetry_proofs}.
Consequently, to reproduce the classical convolutional inductive bias under amplitude-based encoding, we must construct layers that commute with $T$ (or with $T_x,T_y$) rather than with $S$.

\subsection{MERA-QCNN as a QCS-equivariant architecture}\label{sec:mera_qcnn}

Many QCNN architectures are inspired by multiscale tensor-network structures such as MERA~\cite{Cong2019}.
Operationally, they apply a repeated local gate pattern along a line (or lattice) of qubits and interleave it with pooling implemented by measurements and classical control.
Such templates can be made QCS-equivariant by imposing cyclic weight sharing together with periodic boundary conditions, a placement pattern containing the full cyclic orbit of each local block, and pooling/readout rules compatible with the same qubit permutation.
Without these additional choices, a standard brickwork layout generally commutes only with the subgroup preserved by its tiling and pooling pattern, rather than with the full one-site cyclic shift $S$.

For pixel-to-qubit encodings, QCS equivariance can indeed serve as a translation bias.
For address/amplitude encodings, however, Lemma~\ref{lem:mismatch} shows that QCS does not generally coincide with pixel translation (PCS).
Thus, a QCS-equivariant layout is not guaranteed to be pixel-translation-equivariant for encoded images.
This motivates our approach: enforce commutation with the encoding-induced shift $T$ directly.

\subsection{Characterization of PCS-equivariant quantum convolution layers}\label{sec:pcs_layer}

We construct a quantum analogue of a convolutional linear layer for address encoding by requiring commutation with pixel translations.
We start in one dimension to keep notation parallel to Fig.~\ref{fig:general_conv}, and then extend to two and three dimensions.

Consider the Hilbert space
\[
\mathcal{H} = \mathcal{H}_{\mathrm{idx}} \otimes \mathcal{H}_{\mathrm{feat}},
\qquad
\mathcal{H}_{\mathrm{idx}} \cong \mathbb{C}^{N},
\]
where $\mathcal{H}_{\mathrm{idx}}$ is spanned by $\{\ket{j}\}_{j\in\mathbb{Z}_N}$ and $\mathcal{H}_{\mathrm{feat}}$ collects the color/feature qubits.
A unitary layer $U$ is PCS-equivariant in $1$D if
\begin{equation}\label{eq:pcs_commute_1d}
U\,(T\otimes I) = (T\otimes I)\,U,
\end{equation}
where $T$ is defined in \eqref{eq:def_T_1d}.
In $2$D we require commutation with both generators:
\begin{equation}\label{eq:pcs_commute_2d}
U\,(T_x\otimes I) = (T_x\otimes I)\,U,
\qquad
U\,(T_y\otimes I) = (T_y\otimes I)\,U,
\end{equation}
with $T_x,T_y$ defined in \eqref{eq:def_T_2d}.
In $3$D one adds the analogous condition for $T_z$.

The commutation relations \eqref{eq:pcs_commute_1d}-\eqref{eq:pcs_commute_2d} are the direct quantum counterpart of the classical condition $T_k A = A T_k$ from Sec.~\ref{sec:classical_convolution}.
They fix the allowed form of $U$.

\begin{theorem}[Fourier-multiplexer form of PCS-equivariant unitaries]\label{thm:pcs_fourier}
Let $T$ be the cyclic shift on $\mathcal{H}_{\mathrm{idx}}$ defined by \eqref{eq:def_T_1d}.
A unitary $U$ on $\mathcal{H}_{\mathrm{idx}}\otimes\mathcal{H}_{\mathrm{feat}}$ satisfies \eqref{eq:pcs_commute_1d} if and only if it can be written as
\begin{equation}\label{eq:pcs_fourier_form_1d}
U = (F_N^\dagger \otimes I)\, \mathcal{B}\, (F_N \otimes I),
\end{equation}
where $F_N$ is the $N$-point quantum Fourier transform on the index register and $\mathcal{B}$ is block-diagonal in the Fourier basis,
\begin{equation}\label{eq:blockdiag_B_1d}
\mathcal{B} = \bigoplus_{k=0}^{N-1} U_k,
\end{equation}
with arbitrary unitaries $U_k$ acting on $\mathcal{H}_{\mathrm{feat}}$.
Equivalently, $\mathcal{B}$ is a multiplexer: it applies a feature transformation $U_k$ conditioned on the Fourier mode $k$.
In two dimensions, $U$ commutes with both $T_x$ and $T_y$ if and only if
\begin{equation}\label{eq:pcs_fourier_form_2d}
U = ((F_{N_x}\!
\otimes\!
F_{N_y})^\dagger \otimes I)\, \mathcal{B}\, ((F_{N_x}\!
\otimes\!
F_{N_y}) \otimes I),
\end{equation}
where $\mathcal{B}$ is block-diagonal over the pair of modes $(k_x,k_y)$:
\begin{equation}\label{eq:blockdiag_B_2d}
\mathcal{B} = \bigoplus_{k_x=0}^{N_x-1}\bigoplus_{k_y=0}^{N_y-1} U_{k_x,k_y}.
\end{equation}
The extension to three dimensions is obtained by replacing $F_{N_x}\otimes F_{N_y}$ with $F_{N_x}\otimes F_{N_y}\otimes F_{N_z}$ and indexing blocks by $(k_x,k_y,k_z)$.
\end{theorem}

The proof idea is as follows.
The shift operator $T$ is diagonal in the Fourier basis: $F_N T F_N^\dagger = \diag(\omega^k)$.
Hence \eqref{eq:pcs_commute_1d} is equivalent to commutation with a non-degenerate diagonal operator on the index register, which forces $U$ to preserve each Fourier eigenspace.
This yields the block-diagonal structure \eqref{eq:blockdiag_B_1d} in the Fourier basis, and \eqref{eq:pcs_fourier_form_1d} follows by conjugation.
In two (and three) dimensions, the commuting family $\{T_x,T_y(,T_z)\}$ is simultaneously diagonalized by the tensor-product Fourier transform, yielding \eqref{eq:pcs_fourier_form_2d}-\eqref{eq:blockdiag_B_2d}.
The full proof, including the projector argument for the commutant, is given in Supplemental Sec.~\ref{_sec:supp_symmetry_proofs}.

Theorem~\ref{thm:pcs_fourier} is the quantum analogue of the classical Fourier characterization of convolution: translation-equivariant linear layers become block-diagonal in the Fourier representation, with each mode independently mixing channels.
The Fourier-based construction is therefore not an ad hoc ansatz: once the encoded pixel-shift operator is chosen as the symmetry, it is the corresponding unitary commutant.

Operationally, a PCS convolutional layer maps the spatial index registers to the Fourier basis with the QFT, applies a multiplexer $\mathcal{B}$ that performs a mode-dependent transformation on feature qubits, and then maps the index registers back with the IQFT.
An explicit two-layer example is shown in Fig.~\ref{fig:2_layers}.
A schematic decomposition of $\mathcal{B}$ into controlled gates is shown in Supplemental Fig.~\ref{_fig:B_decomposition}.

\begin{figure*}[t] 
\centering
\includegraphics[width=0.8\textwidth]{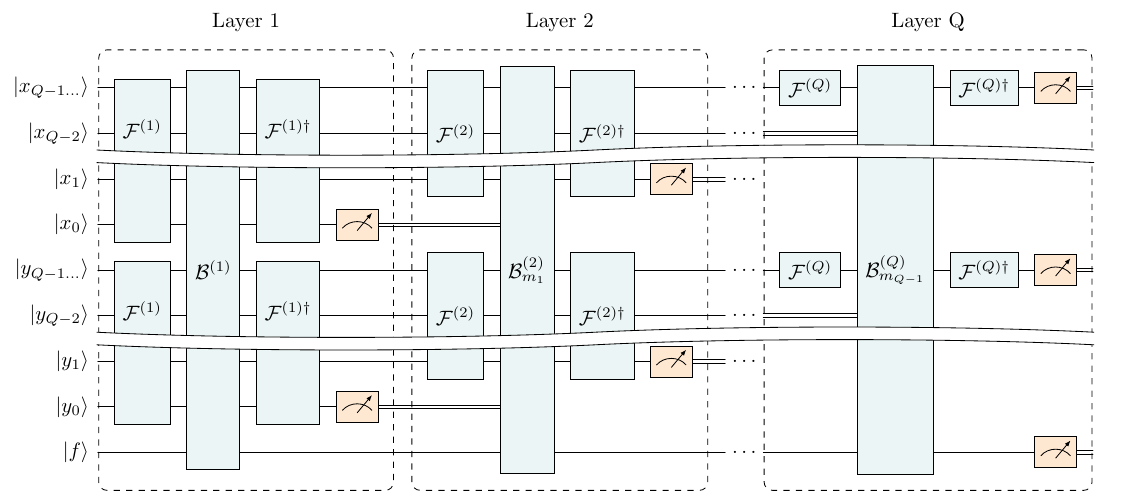}
\caption{Schematic of a multilayer PCS-QCNN.
Each non-final layer maps the active index registers to the Fourier basis with $\mathcal{F}^{(\ell)}$, applies the mode multiplexer $\mathcal{B}^{(\ell)}$, and maps back with $\mathcal{F}^{(\ell)\dagger}$ before pooling measurements are performed on the least significant qubits of the active computational index registers.
The resulting classical outcome
$m_{\ell}$ is recorded and used as a classical control for the next multiplexer block.
After repeated pooling steps (indicated by $\cdots$), the final layer acts on a smaller active index register and is followed by full measurement of the remaining quantum output.
Here $x_{Q-1\ldots}$ and $y_{Q-1\ldots}$ denote all remaining index qubits with labels $Q-1,Q,\ldots$.}
\label{fig:2_layers}
\end{figure*}

\subsection{From linear PCS layers to a deep QCNN: measurement-induced pooling and deferred conditioning}\label{sec:pooling_features}

A composition of unitary PCS-equivariant layers remains a unitary transformation.
Therefore, if we only stack layers of the form \eqref{eq:pcs_fourier_form_2d}, the overall map remains linear in the state amplitudes.
To build a deep architecture analogous to classical CNNs, we introduce a measurement-induced pooling stage between convolutional blocks.
In classical networks this role is played by pointwise activations and pooling.
Here pooling is implemented by measurement and classical control.
At the density-matrix level the process remains linear and CPTP; non-unitarity and branch conditioning enter through measurement, and normalized postselected branch states are nonlinear functions of the pre-measurement state.

We adopt a pooling mechanism inspired by MERA-QCNN but adapted to the Fourier-indexed PCS setting.
After a convolutional block we measure the least significant active computational-index qubit per spatial axis (in $2$D: one in $\mathrm{Reg}\,x$ and one in $\mathrm{Reg}\,y$) in the computational ($Z$) basis, obtaining a classical bit string $m_{\ell}\in\{0,1\}^{d}$.
In the reduced Fourier junction of Sec.~\ref{sec:fourier_interface}, this spatial-parity measurement is represented in Fourier space by mixing the alias pair $k=q$ and $k=q+N/2$ for each surviving coarse mode $q$.
Conditioned on $m_{\ell}$, we apply the next trainable block.
This produces two effects at once: it introduces a measurement-conditioned branch structure with non-unitary branch updates, and it reduces spatial resolution by removing the measured fine-scale index qubits from subsequent processing.
This is the quantum analogue of classical pooling: information carried by fine spatial scales is aggregated into coarser degrees of freedom that remain available to later layers.
The scope of exact equivariance is therefore layerwise: each unitary block is exactly PCS-equivariant on the active index register at that depth.
Pooling yields a multiscale/coarse equivariance structure rather than exact one-pixel equivariance of the entire readout/classifier pipeline.
After pooling, the next block is constrained to commute with the induced translation on the reduced register.
In particular, after one pooling step a unit shift on the coarse lattice corresponds to a shift by $2$ sites on the previous lattice (and by $2^r$ sites after $r$ pooling steps, relative to the original resolution).

\label{sec:new_features}
In classical CNNs, decreasing spatial resolution is often accompanied by an increase in the number of channels.
In the formal model used in this paper (Sec.~\ref{subsec:pcs_bp_model}), we keep the feature-register size fixed within each architecture instance: $n_f$ (equivalently $D_f=2^{n_f}$) does not vary with depth.
Allowing the feature register to grow across layers is a possible extension, but it is not part of the canonical definition used in the gradient-accounting discussion below.

With these ingredients, one level of PCS-QCNN in $2$D can be summarized as
\begin{equation}
\mathcal{L}_{\ell}
=
\mathcal{C}^{(\ell+1)}_{m_{\ell}}
\circ
\mathcal{M}^{(\ell)}_{\mathrm{pool}}
\circ
(\mathcal{F}_x^{(\ell)\dagger} \otimes \mathcal{F}_y^{(\ell)\dagger})
\circ
\mathcal{B}^{(\ell)}
\circ
(\mathcal{F}_x^{(\ell)} \otimes \mathcal{F}_y^{(\ell)}).
\end{equation}
Here, $\mathcal{F}_x^{(\ell)}$ and $\mathcal{F}_y^{(\ell)}$ denote QFTs on the active index registers at depth $\ell$ (so their register size changes with pooling depth).
Equivalently, in the notation of Sec.~\ref{subsec:pcs_bp_model}, $\mathcal{F}^{(\ell)}=\mathcal{F}_x^{(\ell)}\otimes \mathcal{F}_y^{(\ell)}$.
Also, $\mathcal{M}^{(\ell)}_{\mathrm{pool}}$ measures one pooling qubit per spatial axis and outputs $m_{\ell}\in\{0,1\}^{d}$, while $\mathcal{C}^{(\ell+1)}_{m_{\ell}}$ denotes conditional selection of the next trainable block.
The full network is obtained by repeating this pattern several times, followed by a readout stage.

\subsection{Fourier cancellation at the interface of PCS-QCNN layers}
\label{sec:fourier_interface}

A practical advantage of the Fourier-based PCS construction is that intermediate QFT blocks do not have to be re-applied from scratch at every depth.
Consecutive PCS-equivariant layers contain the inverse Fourier transform at the end of one layer next to the forward Fourier transform at the start of the next layer.
Pooling measures (and then discards) the parity bit in the computational-index split $r=2q+s$, while the next layer's QFT acts only on the surviving coarse index $q$.
As a result, the interface formed by the IQFT, the pooling measurement, and the next-layer QFT collapses to a fixed, parameter-free junction from fine Fourier modes to coarse Fourier modes.

Operationally, this junction can be implemented using only local gates in the Fourier alias ordering $k=q+\sigma N/2$: a Hadamard on each alias selector $\sigma$, its measurement, and a conditional diagonal phase-gradient on the remaining index register.
This yields a cleaner multilayer description in which the only trainable operations are the Fourier-mode multiplexer blocks $\mathcal{B}^{(\ell)}$, while the inter-layer wiring is fixed.
This simplification is also useful in the trainability analysis, since it isolates the parameter dependence to the multiplexer blocks.
A full derivation and an explicit circuit diagram for the reduced junction are given in Supplemental Sec.~\ref{_sec:supp_fourier_interface} (see also Supplemental Fig.~\ref{_fig:fourier_junction_reduction}).

\subsection{Readout, hybrid classifier, and the role of shots}\label{sec:readout}

After the quantum convolutional part, classical CNNs typically use dense layers for classification.
We do not implement a quantum fully connected head \cite{Lu2021, Li2020, Li2022}; the architecture is restricted to encoding-aligned quantum convolution for amplitude-based encodings.
Instead, we use an architecturally simple classical classifier operating on the probability vector extracted from the quantum state by measurement.
In the notation of Sec.~\ref{subsec:pcs_bp_model}, for each input $x$ the quantum core produces
\begin{equation}\label{eq:classic_out}
p_{\Theta}(\cdot\mid x)\in\Delta^{D_{\mathrm{out}}},
\end{equation}
which is then mapped to M-class probabilities by
\begin{equation}
q(x)=\mathrm{softmax}\!
\big(W\,p_{\Theta}(\cdot\mid x)+b\big)\in\Delta^M.
\end{equation}
The training loss is the cross-entropy in Eq.~\eqref{eq:loss_def}.
This final linear head is not symmetry-tied across readout coordinates.
Thus, the exact PCS-equivariance statement in this section concerns the unitary quantum convolutional blocks on their active registers and the induced coarse-register action after pooling.

On physical hardware, $p_{\Theta}(\cdot\mid x)$ is estimated from repeated measurements (shots).
Therefore the shot budget $N_{\mathrm{shot}}$ becomes an explicit hyperparameter of the model: it controls the variance of the readout and affects both training and inference.
In the numerical experiments, idealized performance uses the exact distribution, and the finite-shot regime is analyzed separately in Sec.~\ref{sec:results}.

\subsection{Idealized benchmark model and resource scaling}\label{sec:complexity}

Resource scaling is central for this construction because the reported benchmark model is not a gate-compiled hardware circuit.
For an $N\times N$ grayscale image under address encoding, the index registers require $2\lceil \log_2 N\rceil$ qubits, and the feature register contributes $n_f$ qubits.
Circuit depth is dominated by the forward and inverse QFTs on the index registers and by the synthesis of the Fourier-mode multiplexer $\mathcal{B}$.
For a standard (textbook) decomposition, an $m$-qubit QFT uses $O(m^2)$ elementary gates (in particular, $m(m-1)/2$ controlled-phase gates) and has depth $O(m^2)$ without additional parallelization~\cite{Nielsen_Chuang_2010,QFT}.
In our architecture QFT acts separately on the $x$ and $y$ index registers, so this contribution remains modest across the resolutions considered here (up to $32\times32$, i.e., $n_x=n_y\le 5$).
The main bottleneck is the multiplexer: compiling a fully general mode-dependent block-diagonal operator with $n_c$ control qubits and $n_f$ target (feature) qubits requires a number of elementary gates that grows exponentially in $n_c$ in the worst case~\cite{Bergholm_2005,Shende_multiplexer}.
The benchmarks in this paper apply the multiplexer as an ideal block-diagonal operator at the statevector level.
Practical NISQ realizations require additional structure (e.g., parameter sharing or a low-rank / low-depth ansatz for the mode blocks), which defines a different model class and changes expressivity/accuracy trade-offs.
Additional scaling discussion is provided in Supplemental Sec.~\ref{_sec:supp_complexity}.

\section{Gradient accounting and trainability diagnostics}
\label{sec:pcs_bp_theory}

Deep variational quantum models can suffer from barren plateaus, where gradients become too small for practical training~\cite{McClean2018}.
For the fully general Fourier-multiplexer PCS-QCNN, there is an additional elementary issue: the number of trainable scalar parameters can grow exponentially with the number of active index qubits.
This makes coordinate-wise gradient statements delicate.
Even if the total gradient norm is perfectly visible, the gradient energy per parameter can become exponentially small simply because the same amount of signal is distributed over exponentially many independent Fourier-mode blocks.

The resulting average per-coordinate quantity is an accounting baseline, whereas practical trainability is governed by the gradient of the loss actually optimized.
Accordingly, Sec.~\ref{sec:results_diagnostics} reports the corresponding gradient norms directly.

For a dataset or minibatch \(\mathcal{D}=\{(x_i,c_i)\}_{i=1}^{N}\), write
\begin{equation}
g_i\coloneqq \nabla_{\Theta}\mathcal{L}(\Theta,W,b\mid x_i,c_i).
\end{equation}
The optimizer for the averaged loss sees the empirical-loss gradient
\begin{equation}
G_{\mathcal{D}}
\coloneqq
\nabla_{\Theta}\frac{1}{N}\sum_{i=1}^{N}\mathcal{L}(\Theta,W,b\mid x_i,c_i)
=\frac{1}{N}\sum_{i=1}^{N}g_i ,
\label{eq:empirical_gradient_def}
\end{equation}
and the training-scale quantity is \(\|G_{\mathcal{D}}\|_2\), with no division by the number of quantum coordinates.
We also report the per-sample RMS diagnostic
\begin{equation}
R_{\mathcal{D}}
\coloneqq
\left(\frac{1}{N}\sum_{i=1}^{N}\|g_i\|_2^2\right)^{1/2},
\label{eq:per_sample_rms_def}
\end{equation}
which measures the typical single-sample gradient before cancellations between samples.
These two quantities answer different questions: \(R_{\mathcal{D}}\) can remain large even when the averaged training gradient \(\|G_{\mathcal{D}}\|_2\) is small because the sample gradients point in different directions.

\subsection{Depth-scaling accounting}
\label{subsec:pcs_bp_model}

We use the same hybrid model as in Sec.~\ref{chapter:arch}.
For the accounting statement, the depth-scaling family keeps
\begin{equation}
 n_l,\ d,\ D_f,\ M\ \text{fixed},\qquad n_{\mathrm{idx}}=n_l+Q-1,
 \label{eq:depth_scaling_regime}
\end{equation}
where \(Q\) is the number of PCS-QCNN layers, \(d\) is the number of spatial axes, \(n_l\) is the number of index qubits per axis left after pooling, \(D_f=2^{n_f}\) is the feature-register dimension, and \(M\) is the number of classes.
The final readout dimension is therefore fixed:
\begin{equation}
D_{\mathrm{idx}}\coloneqq 2^{d n_l},
\qquad
D_{\mathrm{out}}\coloneqq D_{\mathrm{idx}}D_f .
\end{equation}
At layer \(\ell\), the active index width per axis is
\begin{equation}
 n_{\ell}=n_{\mathrm{idx}}-\ell+1=n_l+Q-\ell.
\end{equation}
With \(\mathcal{F}^{(\ell)}\) denoting the \(d\)-fold QFT on the active index register, a conditional PCS-equivariant layer has the form
\begin{equation}
  U^{(\ell)}(m_{\ell-1}) \coloneqq \big(\mathcal{F}^{(\ell)\dagger}\otimes \one\big)\,
  \mathcal{B}^{(\ell)}(m_{\ell-1})\,
  \big(\mathcal{F}^{(\ell)}\otimes \one\big),
  \label{eq:U_layer_def}
\end{equation}
where
\begin{equation}
\mathcal{B}^{(\ell)}(m_{\ell-1})=
\sum_{k\in[2^{d n_{\ell}}]} \ket{k}\!\bra{k}\otimes V^{(\ell)}_k(m_{\ell-1})
\end{equation}
is a Fourier-mode multiplexer.
As in the benchmark architecture, later multiplexers are conditioned only on the most recently pooled \(d\)-bit outcome, so the branch multiplicity is
\begin{equation}
 b_1=1,\qquad b_{\ell}=2^d\quad (\ell\ge 2).
\end{equation}
The complete measurement-and-feedforward quantum process defines a probability distribution
\begin{equation}
p_{\Theta}(z\mid x),\qquad z\in[D_{\mathrm{out}}],
\label{eq:p_def}
\end{equation}
where $\Theta$ denotes all quantum parameters.
The linear-softmax head maps $p_{\Theta}(\cdot\mid x)$ to class probabilities by
\begin{equation}
q(x)=\mathrm{softmax}\big(Wp_{\Theta}(\cdot\mid x)+b\big)\in\Delta^M.
\end{equation}
For label $c\in\{1,\dots,M\}$ the loss is
\begin{equation}
\mathcal{L}(\Theta,W,b\mid x,c)\coloneqq -\log q_c.
\label{eq:loss_def}
\end{equation}

For the full Pauli-exponential block ansatz in Eq.~\eqref{eq:pauli_exp_param}, each mode-wise feature block has \(p_{\mathrm{blk}}=4^{n_f}\) real scalar parameters.
The quantum-core parameter count is therefore
\begin{equation}
P_Q
= p_{\mathrm{blk}}\sum_{\ell=1}^{Q} b_{\ell}\,2^{d(n_l+Q-\ell)}.
\label{eq:PQ_count}
\end{equation}
In particular,
\begin{equation}
P_Q\ge p_{\mathrm{blk}}\,2^{d(n_l+Q-1)},
\label{eq:PQ_lower}
\end{equation}
so $P_Q$ grows exponentially in $Q$ at fixed $n_l,d,n_f$.
Let
$G_Q=\nabla_{\Theta_Q}\mathcal{L}(\Theta_Q,W,b\mid x,c)\in\mathbb{R}^{P_Q}$
be the gradient with respect to the quantum-core parameters for any fixed sample, minibatch-average loss, or empirical-risk average.
In the initialization protocol used for the benchmark, all quantum Pauli coefficients in Eq.~\eqref{eq:pauli_exp_param} are sampled independently from \(\mathrm{Unif}(0,2\pi)\).
The classical head is initialized independently by the fan-in uniform rule
\begin{equation}
  \begin{aligned}
  W_{az}&\sim\mathrm{Unif}\!\left(-D_{\mathrm{out}}^{-1/2},D_{\mathrm{out}}^{-1/2}\right),
  & a&\in[M],\\
  b_a&\sim\mathrm{Unif}\!\left(-D_{\mathrm{out}}^{-1/2},D_{\mathrm{out}}^{-1/2}\right),
  & z&\in[D_{\mathrm{out}}],
  \end{aligned}
\end{equation}
independently over all entries, where \(D_{\mathrm{out}}=2^{d n_l}D_f\) is fixed in the depth-scaling regime.
The elementary sensitivity estimate in Supplemental Sec.~\ref{_sec:supp_trainability_full} gives
\begin{equation}
\mathbb{E}\|G_Q\|_2^2
\le
\frac{8M}{3}p_{\mathrm{blk}}Q ,
\end{equation}
where the expectation is over this independent initialization.
Dividing by~\eqref{eq:PQ_lower} yields
\begin{equation}
  \mathbb{E}\!\left[\frac{1}{P_Q}\|G_Q\|_2^2\right]
  \le
  \frac{8M Q}{3\,2^{d(n_l+Q-1)}} .
  \label{eq:coord_count_bound}
\end{equation}
For the ten-class MNIST heads used in the experiments, the numerator constant is \(80Q/3\).
Thus the average per-coordinate second moment is upper-bounded by a quantity scaling as \(Q2^{-dQ}\) in this formal depth-scaling family.
This establishes the parameter-count accounting baseline.
It follows from the exponential growth of \(P_Q\), while aggregate gradient norms are measured separately in Sec.~\ref{sec:results_diagnostics} through \(\|G_{\mathcal{D}}\|_2\) and \(R_{\mathcal{D}}\).

\section{Benchmark choice: MNIST for convolutional inductive bias}\label{sec:mnist_motivation}

MNIST provides a widely used and computationally tractable benchmark~\cite{LeCunMNIST}.
We use it to test translation-equivariant inductive bias under limited qubit and depth budgets.
Because quantum-machine-learning benchmarks require careful matched controls~\cite{SchuldKilloran2022,Cerezo2022}, the MNIST protocol below is used to isolate the convolution-sensitive part of the comparison.

The practical benefit of convolution is regime-dependent.
On non-translated full-data MNIST, dense models can narrow the gap to CNNs, which makes the convolutional effect harder to diagnose.
Supplemental Fig.~\ref{_fig:supp_full_mnist_classical_baselines} shows this directly for the classical controls used here: on the full standard MNIST split resized to $32\times32$ without translations, the CNN and MLP remain much closer than on translated-MNIST ($99.15\%$ vs $96.74\%$ final mean test accuracy).
To make the translation-sensitive inductive bias explicit, we use a translated-MNIST regime: each digit is resized to $16\times16$, embedded in a $32\times32$ canvas, and translated by a random integer offset of at most $8$ pixels per axis.
We train on a balanced subset of $1000$ examples per class and check explicitly that a convolutional classical CNN reaches higher accuracy than a classical MLP on this translated task.
The same translated benchmark is then used for the PCS-QCNN versus matched non-PCS RBC-QCNN comparison in Fig.~\ref{fig:small_ds_clas}(b).

Quantum image models are also constrained by input dimension.
Even MNIST's native $28\times 28$ resolution is too costly for straightforward state preparation and fully general multiplexer implementations on near-term hardware, so classical preprocessing remains practically important.
In the reported experiments we work with power-of-two canvases and a small set of controlled preprocessing choices: the translated benchmark uses digits resized to $16\times16$ and placed on a $32\times32$ canvas, while the full-MNIST sweep in Secs.~\ref{sec:benchmark} and \ref{sec:results} uses four settings: direct preprocessing to $8\times8$, direct preprocessing to $16\times16$, centered embedding of native $28\times28$ images in a $32\times32$ canvas, and direct preprocessing to $32\times32$.
Alternative compression methods such as PCA can reduce dimension more aggressively but may suppress spatial correlations and are therefore less suitable for isolating convolutional inductive bias.

Finally, a large body of QCNN-related work has used MNIST under heterogeneous settings (binary vs multi-class tasks, varying downscaling, and different train/test protocols).
We use a fully specified protocol and include a comparable-budget classical reference effect (CNN vs MLP) to isolate architectural conclusions about quantum convolution.
For a literature map of QCNN-like approaches on MNIST and representative reported settings/accuracies, see Supplemental Sec.~\ref{_sec:supp_mnist_context}.

\section{Benchmark}\label{sec:benchmark}

The benchmark protocol fixes the data preprocessing, hybrid model (encoding, quantum core, classical head), training and evaluation procedures, and simulation modes (statevector vs finite-shot sampling).
The MNIST-specific motivation is given separately in Sec.~\ref{sec:mnist_motivation}; the present section gives the full experimental specification.

\subsection{Data, splits, and preprocessing}\label{sec:benchmark_data}

We use two MNIST regimes.
For Figs.~\ref{fig:small_ds_clas} and~\ref{fig:diff_models}(a), we use a translated-MNIST benchmark built from a balanced subset of $1000$ training examples per class; the standard MNIST test split membership is kept unchanged.
Within this translated regime, both train and test images are resized from $28\times28$ to $16\times16$ by bilinear interpolation, placed on a zero-filled $32\times32$ canvas, and translated by an independently sampled integer offset with maximum magnitude $8$ pixels along each axis.
The prepared translated train and test sets are fixed within each independent run and reused for all training epochs and test evaluations.
The PCS symmetry is formulated on a periodic lattice to obtain an exact finite group action, while this benchmark uses only the corresponding non-wrapping translation regime: the resized digit patch remains inside the $32\times32$ canvas, so no occupied pixel is transported across the boundary.
Thus, the translated-MNIST split probes realistic bounded translation stability on the same discrete lattice without requiring the dataset to contain complete cyclic orbits.
These translated-MNIST runs use train/test batch sizes $256/1600$.

For Fig.~\ref{fig:diff_models}(b) and all finite-shot and diagnostic follow-up plots, we use the full-MNIST split ($60{,}000$ training images and $10{,}000$ test images) with no random translations.
The evaluated full-MNIST preprocessing settings are direct preprocessing to $8\times8$, direct preprocessing to $16\times16$, centered embedding of the native $28\times28$ images in a $32\times32$ canvas, and direct preprocessing to $32\times32$.
These full-MNIST runs use train/test batch sizes $512/16000$.
In all cases the preprocessing protocol outputs grayscale tensors normalized to $[0,1]$.

\subsection{Encoding and state preparation}\label{sec:benchmark_encoding}

Each input image is encoded by a FRQI-like brightness map~\cite{Le2011}.
Let $x_{u,v}\in[0,1]$ denote the preprocessed grayscale value at pixel $(u,v)$.
The encoder maps it affinely to an angle
\begin{equation}
    p_{u,v}=a+(b-a)x_{u,v},
\end{equation}
and writes the local feature state from Eq.~\eqref{eq:color_map_impl} into the least-significant feature qubit.
For \(n_f>1\), all remaining auxiliary feature qubits are initialized in \(\ket{0}^{\otimes(n_f-1)}\); in our feature-register ordering the local state is \(\ket{0}^{\otimes(n_f-1)}\otimes\ket{\phi_{u,v}}\), with \(\ket{\phi_{u,v}}\) on the least-significant feature qubit.
For the reported PCS-QCNN experiments we use the brightness interval $(a,b)=(0,\pi)$.
We choose this a priori as a fixed one-period convention for the local real FRQI color map, not by optimizing it on the reported test set.
Locally, $\ket{\phi(p+\pi)}=-\ket{\phi(p)}$, so intervals longer than $\pi$ revisit the same one-qubit rays.
In the coherent FRQI image superposition, however, such signs can become address-dependent relative phases rather than a global phase of the full image state.
A separate low-data brightness-range sweep varies the upper endpoint $b$ while fixing $a=0$ and is reported only as a post-hoc sensitivity check of this convention (Supplemental Sec.~\ref{_sec:brightness_sweep} and Fig.~\ref{_fig:brightness_sweep}).
Thus an image on an $N_x\times N_y$ canvas is represented using $\log_2 N_x+\log_2 N_y$ index qubits together with $n_f$ feature-register qubits (the grayscale/color qubit plus optional auxiliary feature qubits).

The input state is initialized directly as a dense tensor rather than by compiling an explicit state-preparation circuit.
Relative to the conceptual normalized form in Eq.~\eqref{eq:frqi_like}, the numerical realization used here omits the global $1/\sqrt{N_xN_y}$ factor during initialization.
Encoded sample norms are therefore $\sqrt{N_xN_y}$ rather than $1$.
The measurement layer compensates for this by dividing the final marginals by the known overall spatial normalization factor, so the classifier still receives normalized readout probabilities.
Approximate and structured FRQI loading circuits remain relevant for future hardware-oriented implementations~\cite{Lu2021,Khoshaman2018,Oseledets}.

\subsection{Models under comparison and hyperparameters}\label{sec:benchmark_models}

We evaluate the proposed translation-symmetry-aligned quantum convolutional architecture using the complete hybrid model shown in Fig.~\ref{fig:full_scheme}.
The model consists of a FRQI-like encoder, a PCS-QCNN quantum core, a marginal measurement stage, an optional finite-shot sampling layer, and a linear-softmax head.
The head applies a single linear map of size $(D_{\mathrm{out}},10)$, where $D_{\mathrm{out}}$ is the dimension of the measured probability vector passed to the classifier.
The resulting logits are converted to probabilities by a softmax and trained with the cross-entropy loss.

\begin{figure}
    \centering
    \includegraphics[width=\columnwidth]{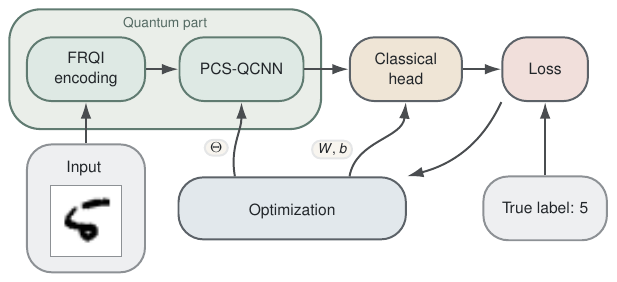}
    \caption{Overview of the hybrid model used in our experiments. The Encoding block prepares an FRQI-like image state on the spatial index registers and the feature register. The PCS-QCNN block applies one or more PCS-equivariant quantum layers with measurement-induced pooling implemented through deferred conditioning. The final measured probability tensor (or its finite-shot histogram estimate at evaluation time) is fed into a linear-softmax head that produces $10$-class logits.}
\label{fig:full_scheme}
\end{figure}

We begin with the quantum part (PCS-QCNN).
A model of depth $Q$ is obtained by composing $Q$ PCS-equivariant layers in the sense of Sec.~\ref{chapter:arch}.
At layer $\ell$, the active index registers are mapped to the Fourier basis, processed by a trainable Fourier-mode multiplexer, and returned to the computational basis; after each non-final layer, one active index qubit per spatial axis is pooled via deferred measurement and moved into an explicit condition register, while the final layer is followed directly by readout.
Later multiplexers are conditioned on the most recently pooled $x$- and $y$-bits.
In the reported results we use the reduced Fourier junction, so each explicit interface consisting of an IQFT, pooling, and the next-layer QFT is replaced by the analytically equivalent fixed junction described in Sec.~\ref{sec:fourier_interface}.
At the density-matrix level the resulting process is linear and CPTP.
In an individual-branch statevector description, the normalized postselected branch states are nonlinear functions of the pre-measurement state.

The reported benchmarks use two PCS-QCNN parameter families.
Figure~\ref{fig:small_ds_clas}(b) fixes a translated-MNIST model with $Q=3$ and $n_f=2$ on a $32\times32$ canvas.
Figure~\ref{fig:diff_models}(a) sweeps $Q\in\{1,2,3,4,5\}$ and $n_f\in\{1,2,3\}$ on the same translated benchmark.
Figure~\ref{fig:diff_models}(b) and all finite-shot follow-up plots use a full-MNIST size sweep with $Q=1$ and $n_f=3$ across four preprocessing settings: direct preprocessing to $8\times8$, direct preprocessing to $16\times16$, centered embedding of native $28\times28$ images in a $32\times32$ canvas, and direct preprocessing to $32\times32$; the follow-up reference model for Figs.~\ref{fig:accuracy_with_shots}, \ref{fig:hist}, and Supplemental Figs.~\ref{_fig:readout_entropy}, \ref{_fig:ellipse}, and~\ref{_fig:decomp} is the medium-resolution $16\times16$ setting.

Let $n_{\mathrm{idx}}=\log_2 N$ denote the number of index qubits per spatial axis for an $N\times N$ preprocessed image.
The total number of logical qubits in the quantum part is
\begin{equation}
n_{\mathrm{tot}}=2n_{\mathrm{idx}}+n_f.
\label{eq:q_tot}
\end{equation}
For the full-readout setting used in all reported figures, the final marginal measurement sums over the explicit condition registers and leaves only the surviving active spatial axes together with the feature register.
Hence
\begin{equation}
n_l=n_{\mathrm{idx}}-Q+1,\qquad n_{\mathrm{meas}}=2n_l+n_f,
\label{eq:q_meas}
\end{equation}
and the classifier input dimension is
\begin{equation}
D_{\mathrm{out}}=2^{n_{\mathrm{meas}}}=2^{2n_l+n_f}.
\label{eq:readout_dim}
\end{equation}
Therefore the classical head size is determined by the number of active qubits left after pooling, not by the total logical register size.

The classifier head is architecturally simple but not always smaller than the quantum core in parameter count.
It is a single biased linear-softmax map $\mathbb{R}^{D_{\mathrm{out}}}\to\mathbb{R}^{10}$, so the classifier contributes exactly $10D_{\mathrm{out}}+10$ trainable parameters.
We denote by $P_Q$ the number of trainable real parameters in the quantum core.
These counts are summarized for the reported sweeps in Supplemental Tables~\ref{_tab:results_architecture} and \ref{_tab:results_size}.

The trainable Fourier-mode blocks also determine the hardware-facing cost model.
In the benchmark model, each mode-dependent feature-register block $U^{(\ell)}_{k_x,k_y}(m)$ is parameterized as a general $U(2^{n_f})$ unitary via the exponential map
\begin{equation}
U^{(\ell)}_{k_x,k_y}(m)\;=\;\exp\!
\Big(i\sum_{\alpha\in\mathcal{P}_{n_f}}\theta^{(\ell)}_{k_x,k_y,\alpha}(m)\,P_{\alpha}\Big),
\label{eq:pauli_exp_param}
\end{equation}
where $\mathcal{P}_{n_f}$ denotes the full $n_f$-qubit Pauli-string basis, including the all-identity string.
The identity coefficient is retained: it is only a phase for a single feature block, but mode- and branch-dependent block phases become relative phases inside the full Fourier multiplexer and therefore are part of the general PCS-equivariant layer.
This gives $4^{n_f}$ real parameters per Fourier mode (and, for $\ell\ge 2$, per pooling branch $m$).
At initialization, every scalar coefficient \(\theta^{(\ell)}_{k_x,k_y,\alpha}(m)\) in Eq.~\eqref{eq:pauli_exp_param} is sampled independently from \(\mathrm{Unif}(0,2\pi)\).
This specifies the coordinate initialization of the exponential ansatz; it is not a Haar draw from \(U(2^{n_f})\).
Since we only use $n_f\le 3$ in the reported benchmarks, each individual mode block acts on at most three qubits and can be synthesized with $O(4^{n_f})$ elementary gates using standard compilation methods; the dominant challenge is multiplexing these blocks across many Fourier modes.
For an $N_x\times N_y$ input, the number of Fourier modes is $N_xN_y$.
At fixed $n_f$, a fully general PCS layer therefore scales linearly in the number of spatial modes, $p_{\rm PCS}=O(N_xN_y)$, whereas an unconstrained dense mode-mixing map scales as $p_{\rm dense}=O((N_xN_y)^2)$.
The benchmark PCS layer uses the general translation-equivariant family: we do not impose finite spatial support or mode-wise weight sharing inside the Fourier blocks.
The corresponding multiplexer $\mathcal{B}^{(\ell)}$ is applied as an explicit block-diagonal unitary in the Fourier basis, i.e., as $\bigoplus_{k_x,k_y}U^{(\ell)}_{k_x,k_y}(m)$ controlled by the active Fourier indices and the selected condition branch.
For simulation, we apply this operator directly as an explicit block-diagonal map.
Rather than explicitly branching on mid-circuit measurement outcomes, pooling and feedforward are implemented through the equivalent deferred-measurement form (keeping pooled qubits as condition-register controls and marginalizing them in the final readout), which is operationally identical to measuring those qubits in the computational basis and classically conditioning the next layer~\cite{Nielsen_Chuang_2010}.

On real NISQ hardware, the dominant cost would be compiling these large multiplexers into one- and two-qubit gates.
In the worst case (no structure shared across modes), standard decompositions of uniformly controlled unitaries require a number of elementary gates that grows exponentially in the number of control qubits (hence polynomially in the number of pixels) and can quickly dominate depth as the image resolution increases~\cite{Bergholm_2005,Shende_multiplexer}.
The controls are matched within each comparison family: the classical CNN and MLP have nearly equal parameter counts, while the PCS-QCNN and RBC-QCNN have identical quantum-core and classifier-head parameter counts.
Hardware-oriented structured/compressed multiplexer designs are discussed in Supplemental Sec.~\ref{_sec:supp_complexity}.

\subsection{Training objective and optimization}\label{sec:benchmark_training}

The training objective is the standard multiclass cross-entropy loss.
For each sample $(x,c)$, the quantum core produces a readout vector $p_{\Theta}(\cdot\mid x)$, the classifier head outputs logits $z(x)=W\,p_{\Theta}(\cdot\mid x)+b$, and class probabilities are obtained as $q(x)=\mathrm{softmax}(z(x))$.
The minibatch loss is therefore
\begin{equation}
\ell = -\frac{1}{B}\sum_{b=1}^{B}\log q_{c_b}(x_b),
\end{equation}
which is the minibatch version of Eq.~\eqref{eq:loss_def}.

All reported models are optimized end-to-end with Adam.
The classical head uses fan-in uniform initialization: with fan-in \(D_{\mathrm{out}}\), both weights and biases are sampled independently from \(\mathrm{Unif}(-D_{\mathrm{out}}^{-1/2},D_{\mathrm{out}}^{-1/2})\).
The classical baselines use learning rate $10^{-2}$, while the hybrid PCS-QCNN models use learning rate $3\times10^{-2}$.
The translated-MNIST fixed runs (Fig.~\ref{fig:small_ds_clas}), the translated-MNIST architecture sweep (Fig.~\ref{fig:diff_models}(a)), and the full-MNIST size sweep (Fig.~\ref{fig:diff_models}(b)) are all trained for $2000$ epochs with test evaluation every $10$ epochs.
The finite-shot follow-up analyses use two subsets of these epochs: Fig.~\ref{fig:accuracy_with_shots} uses epochs $10,100,200,\dots,800$, while Fig.~\ref{fig:hist} uses epochs $100$ and $800$.

\subsection{Simulation modes: statevector vs finite-shot readout}\label{sec:benchmark_simulation_modes}

All experiments are executed on a classical simulator.
We use two simulation modes for the quantum part:

In the infinite-shot (statevector / exact-probability) mode, the simulator computes the exact quantum state (statevector) and the exact output probabilities for the measured qubits.
This removes sampling noise and corresponds to the formal $N_{\mathrm{shot}}\to\infty$ limit.
To avoid ambiguity, all training is performed exclusively in this infinite-shot mode.
This choice isolates architectural effects (symmetry, depth, feature register size, conditioning structure) from the stochasticity of finite-sampling training, and it reflects the fact that stable finite-shot training for larger hybrid models remains computationally expensive.

In the finite-shot mode, the measured probability vector is estimated from a finite number of measurement shots $N_{\mathrm{shot}}$ by sampling from the output distribution.
Equivalently, for each test image with exact readout vector \(p_{\Theta}(\cdot\mid x)\), the classifier receives the empirical frequency vector \(\hat p\) obtained from \(N_{\mathrm{shot}}\) multinomial draws from \(p_{\Theta}(\cdot\mid x)\).
This emulates the fundamental finite-sampling noise of quantum readout even on ideal (noise-free) hardware.
We apply the finite-shot mode only at inference time.
Concretely, for a fixed trained model, we evaluate the classifier on the test set by replacing the exact probability vector with its finite-shot estimate and report accuracy as a function of $N_{\mathrm{shot}}$ (e.g., $128,256,512,1024,2048$ shots).

Thus, throughout the paper, training always uses infinite-shot (exact) quantum readout, whereas inference is reported in both infinite-shot (exact) and finite-shot modes.

\subsection{Baselines and matched controls}\label{sec:benchmark_baselines}

We include two classical references and one matched quantum control around the translated-MNIST benchmark from Sec.~\ref{sec:benchmark_data}.
First, Fig.~\ref{fig:small_ds_clas}(a) compares a convolutional classical CNN against a pure-dense MLP control on translated-MNIST with $16\times16$ digits placed on a $32\times32$ canvas.
The CNN is the four-convolution architecture shown in Supplemental Fig.~\ref{_fig:supp_classical_baselines}(a), and the MLP is the three-hidden-layer dense control shown in Supplemental Fig.~\ref{_fig:supp_classical_baselines}(b).
Their trainable-parameter counts are of the same order ($47{,}034$ for the CNN and $47{,}947$ for the MLP), so the main architectural difference remains the presence or absence of convolutional weight sharing.
Supplemental Fig.~\ref{_fig:supp_full_mnist_classical_baselines} then applies the same pair to the full standard MNIST split resized directly to $32\times32$ without translations, which acts as a control showing that non-translated MNIST alone is a less stringent test of convolutional inductive bias.

Second, Fig.~\ref{fig:small_ds_clas}(b) compares the PCS-QCNN against a matched non-PCS RBC-QCNN on the same translated benchmark.
In this control, each layer replaces the Fourier-basis change and its inverse by a fixed, non-trainable, layer-specific random spatial unitary $R_{\ell}$ and its adjoint $R_{\ell}^{\dagger}$, applied identically on the active $x$- and $y$-registers.
Each $R_{\ell}$ is generated by exponentiating an i.i.d. Gaussian Pauli-basis Hermitian generator of the appropriate active-index dimension.
The control keeps the same depth $Q$, feature-qubit count $n_f$, multiplexer parameter count, pooling schedule, full readout, and classical head as the PCS-QCNN; unlike the reported PCS-QCNN model, it uses explicit pooling rather than the reduced Fourier-junction shortcut.
This preserves the matched architecture and parameter count but generally removes the PCS-equivariance guarantee.
For the fixed translated run used in Fig.~\ref{fig:small_ds_clas}(b), both quantum variants use $Q=3$, $n_f=2$, full readout, and $39{,}434$ trainable parameters in total ($36{,}864$ quantum and $2{,}570$ in the classifier).
The results of these baseline comparisons are reported in Sec.~\ref{sec:results} and summarized numerically in Table~\ref{tab:baseline_comp}.

\subsection{Numerical realization}\label{sec:benchmark_implementation}

All reported results are obtained by dense-tensor statevector simulation.
The encoder, Fourier transforms, reduced Fourier junctions, multiplexers, marginal measurement, finite-shot sampling layer, and classifier are evaluated directly on the full state tensor rather than through a gate-by-gate hardware emulation.
Training uses exact-probability statevector evolution; the finite-shot studies replace the exact readout by multinomial histograms sampled from the same trained quantum output distribution.
The reproducibility package is available in \cite{pcsqcnn_repo}.

\section{Results}\label{sec:results}

The empirical results for the hybrid PCS-QCNN follow the benchmark protocol fixed in Sec.~\ref{sec:benchmark}.
They address the effect of the PCS construction relative to a matched non-PCS RBC-QCNN, the classical CNN/MLP control across translated and non-translated MNIST regimes, infinite-shot behavior across architecture and input-size sweeps, and finite-shot readout as a resource constraint.
As fixed above, the quantum control is a matched non-PCS ablation rather than a QCS baseline, and the finite-shot diagnostics use a single direct-$16\times16$ full-MNIST reference run.
Unless stated otherwise, all models are trained in the infinite-shot (exact-probability) simulator mode; at inference time we report both infinite-shot performance and finite-shot reevaluations obtained by sampling the quantum readout distribution with a fixed shot budget.

\subsection{Infinite-shot inference performance and learning dynamics}\label{sec:results_infinite_shot}

Figure~\ref{fig:diff_models}(a) reports the translated-MNIST architecture sweep under infinite-shot inference.
In this sweep, several multi-layer models exceed the one-layer configurations with the same feature width, but performance is not monotone in depth.
The largest displayed mean occurs at the intermediate setting $Q=3$, $n_f=3$ ($79.15\%$), but because the sweep reports only three-seed means and no uncertainty estimates, we treat this as a descriptive sweep outcome rather than a resolved architecture ranking.
Increasing from $n_f=1$ to wider feature registers often improves performance, although the trend is not monotone at the deepest settings.
Thus, on the translated benchmark, both depth and feature width materially affect performance.

Figure~\ref{fig:diff_models}(b) isolates input resolution on full-MNIST using the canonical $Q=1$, $n_f=3$ PCS-QCNN.
Moving beyond $8\times8$ gives a clear gain: the $8\times8$ model has the lowest final mean ($93.21\%$), while the $16\times16$, $28\times28$-on-$32\times32$, and direct-$32\times32$ models all reach high final accuracies.
Among the displayed preprocessing choices, the direct-$32\times32$ run attains the highest final mean ($98.13\%$).
For a map of QCNN-like results under heterogeneous MNIST settings, see Supplemental Sec.~\ref{_sec:supp_mnist_context}.
A detailed error breakdown for the final direct-$16\times16$ full-MNIST reference model is provided in Supplemental Sec.~\ref{_sec:supp_results_errors}, including a confusion matrix and representative misclassified examples.
The remaining errors are consistent with visual ambiguity between similar digits after aggressive downscaling to a $16\times16$ quantum input.

\begin{figure*}
    \centering
    \begin{overpic}[width=0.49\textwidth]{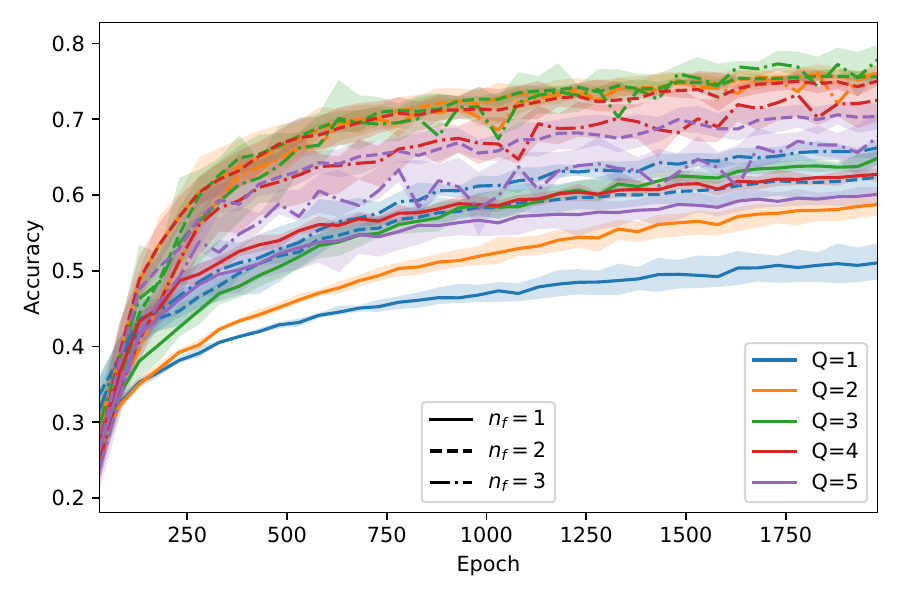}
    \put(3mm,3mm){(a)}
    \end{overpic}
    \begin{overpic}[width=0.49\textwidth]{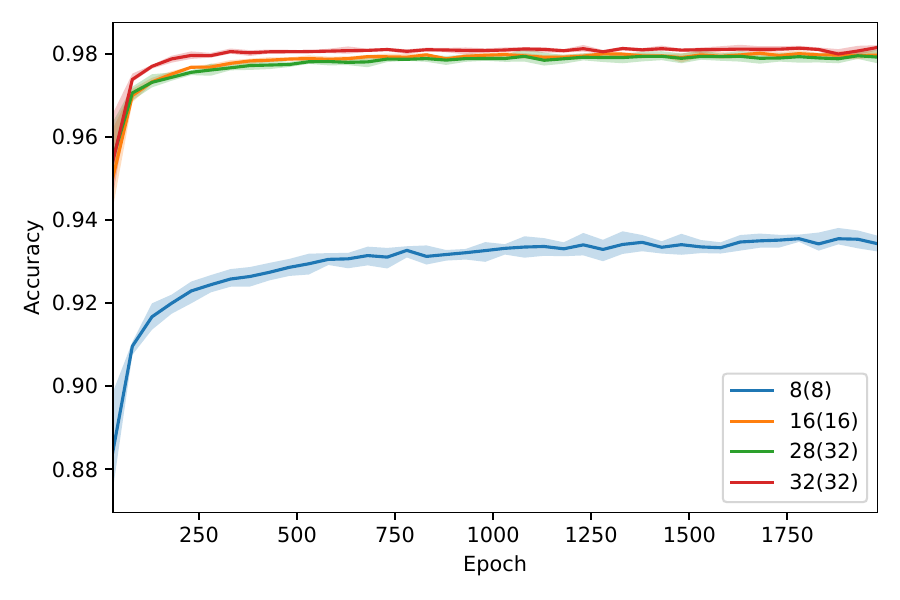}
    \put(3mm,3mm){(b)}
    \end{overpic}
    \caption{Infinite-shot test-accuracy dynamics for the PCS-QCNN sweeps. (a) Translated-MNIST architecture sweep over quantum layers $Q\in\{1,\dots,5\}$ and feature-qubit counts $n_f\in\{1,2,3\}$ with full readout on the translated benchmark with $16\times16$ digits placed on a $32\times32$ canvas; colors encode $Q$ and line styles encode $n_f$. (b) Full-MNIST size sweep for the canonical $Q=1$, $n_f=3$ PCS-QCNN across four preprocessing settings: direct $8\times8$, direct $16\times16$, centered $28\times28$ in a $32\times32$ canvas, and direct $32\times32$. In both panels, solid curves show mean test accuracy and shaded bands show the $25$th--$75$th percentile range over $3$ seeds.}
    \label{fig:diff_models}
\end{figure*}

\subsection{Hyperparameter sweep and parameter accounting}\label{sec:results_sweep}

The translated-MNIST architecture sweep spans $Q\in\{1,2,3,4,5\}$ and $n_f\in\{1,2,3\}$.
In this sweep, depth and feature width are both relevant: wider feature registers often help, and the top observed means occur at intermediate depths, but the dependence is not monotone and the sweep should not be read as a precise architecture ranking.
Exact parameter counts for the reported architecture and size sweeps are summarized in Supplemental Tables~\ref{_tab:results_architecture} and \ref{_tab:results_size}; the corresponding accuracy trends are shown directly in Fig.~\ref{fig:diff_models}.

\subsection{Matched controls on translated-MNIST}\label{sec:results_baselines}

In the translated-MNIST regime from Sec.~\ref{sec:benchmark_data}, we include two comparisons: (i) classical CNN versus classical MLP on translated-MNIST with $16\times16$ digits on a $32\times32$ canvas, and (ii) the PCS-QCNN versus the matched non-PCS RBC-QCNN under the translated hybrid protocol (Sec.~\ref{sec:benchmark_baselines}).
The results are summarized in Table~\ref{tab:baseline_comp}.

Figure~\ref{fig:small_ds_clas}(a) shows a wide separation between the classical CNN and MLP on the translated benchmark: with $16\times16$ digits placed on a $32\times32$ canvas, the CNN reaches a final mean test accuracy of $97.68\%$, whereas the dense MLP reaches $48.93\%$.
Supplemental Fig.~\ref{_fig:supp_full_mnist_classical_baselines} shows that the same classical architectures are much closer on full $32\times32$ MNIST without translations ($99.15\%$ vs $96.74\%$), illustrating that non-translated MNIST alone is a less stringent benchmark for testing convolutional inductive bias.
Within the quantum family, Fig.~\ref{fig:small_ds_clas}(b) shows a $35.08$ percentage-point gap between the PCS-QCNN and the matched non-PCS RBC-QCNN: $75.89\%$ versus $40.82\%$ final mean test accuracy.
Thus, although the fixed $Q=3$, $n_f=2$ PCS-QCNN does not match the classical CNN on this translated benchmark, the comparison provides evidence for a PCS-specific contribution relative to a matched symmetry-broken quantum control.

\begin{table}[t]
\centering
\begin{tabular}{@{}p{0.50\columnwidth}p{0.23\columnwidth}p{0.19\columnwidth}@{}}
\hline
\textbf{Model / baseline} & \textbf{Params} & \textbf{Accuracy (\%)} \\
\hline
Classical CNN & $47{,}034$ & $97.68$ \\
Classical MLP & $47{,}947$ & $48.93$ \\
PCS-QCNN & $39{,}434$ & $75.89$ \\
RBC-QCNN & $39{,}434$ & $40.82$ \\
\hline
\end{tabular}
    \caption{Final mean test accuracies for the translated benchmark entries in Fig.~\ref{fig:small_ds_clas} ($1000$ training examples per class, with digits resized to $16\times16$, placed on a $32\times32$ canvas, and translated by at most $8$ pixels along each axis for the classical and quantum runs). The parameter column reports the total number of trainable parameters. Both quantum variants use $Q=3$, $n_f=2$, and full readout. The reported value in each row is the mean over $3$ seeds.}
\label{tab:baseline_comp}
\end{table}

\subsection{Optimization measurements}\label{sec:results_diagnostics}

The parameter-count estimate in Sec.~\ref{sec:pcs_bp_theory} controls the mean per-coordinate second moment.
The practical trainability question is whether the empirical-loss gradient seen by the optimizer is numerically visible.
We measure two initialization-time gradient diagnostics in a depth-scaling family on $2^Q\times2^Q$ inputs, with fixed post-pooling size \(n_l=1\) and fixed feature register \(n_f=3\).
Figure~\ref{fig:gradient_norms_main} reports the empirical-loss gradient norm \(\|G_{\mathcal{D}}\|_2\) and the per-sample RMS gradient norm \(R_{\mathcal{D}}\), each for all quantum parameters, for the first quantum layer, and for the last quantum layer.
The solid curves isolate the actual averaged training signal, while the dashed curves show the typical single-sample signal before inter-sample cancellation.
Thus the diagnostic separates two possible sources of small gradients: attenuation of the single-sample gradients themselves and cancellation when the empirical loss is averaged over the dataset.
For each depth \(Q\in\{1,\dots,7\}\), the diagnostics are computed at initialization on the same deterministic class-balanced subset \(\mathcal{D}\) of \(N=256\) test samples.
For this diagnostic only, the selected raw test images are resized to \(2^{Q-1}\times2^{Q-1}\), placed on a \(2^Q\times2^Q\) canvas, and translated by fixed integer offsets within the available border; the resulting depth-specific inputs are held fixed across the \(12\) random initializations.
For a fixed random initialization, \(\|G_{\mathcal{D}}\|_2\) is obtained by differentiating the cross-entropy loss averaged over this subset and then taking the Euclidean norm of the selected quantum-parameter coordinates.
By contrast, \(R_{\mathcal{D}}\) is obtained by differentiating each single-sample loss separately, taking the Euclidean norm on the same selected coordinate set, squaring, averaging these squared norms over the \(256\) samples, and finally taking the square root.
The three selected coordinate sets are the full quantum core, the first quantum layer, and the last quantum layer.
Neither diagnostic is divided by the number of coordinates.
The plotted line is the mean over \(12\) independent random initializations, and the shaded band is the corresponding \(25\)th--\(75\)th percentile range.

\begin{figure}[t]
    \centering
    \includegraphics[width=\columnwidth]{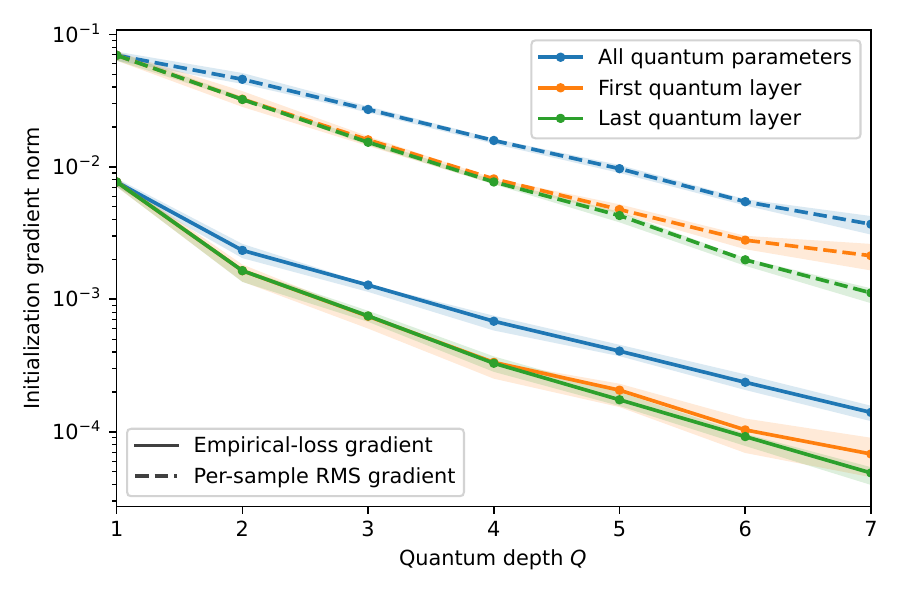}
    \caption{Initialization-time quantum-gradient diagnostics in the depth-scaling family. Solid curves show the empirical-loss gradient norm \(\|G_{\mathcal{D}}\|_2\) and dashed curves show the per-sample RMS gradient norm \(R_{\mathcal{D}}\), defined in Eqs.~\eqref{eq:empirical_gradient_def} and~\eqref{eq:per_sample_rms_def}. Colors indicate whether the selected coordinates are all quantum parameters, the first quantum layer, or the last quantum layer.}
    \label{fig:gradient_norms_main}
\end{figure}

We also monitor the Shannon entropy of the full readout distribution under finite-shot reevaluation in Supplemental Fig.~\ref{_fig:readout_entropy}.
The entropy tracks how broad the sampled output histograms remain as training proceeds, separating finite-shot readout effects from the gradient-scale diagnostics above.

\subsection{Finite-shot inference and a representative train--deploy mismatch}\label{sec:results_shots}

The practical cost of inference on quantum hardware scales with the number of measurement shots used to estimate the readout distribution.
Therefore, shot budget becomes an explicit hyperparameter (Sec.~\ref{sec:benchmark_simulation_modes}).
We evaluate finite-shot inference by sampling from the exact output distribution produced by the statevector simulator and then running the classical head on the estimated probability vector.

Figure~\ref{fig:accuracy_with_shots} reevaluates the representative direct-$16\times16$ full-MNIST reference model ($Q=1$, $n_f=3$) at selected training epochs and several shot budgets.
Each finite-shot point is one stochastic full-test-set pass in which each test image receives one independent \(N_{\mathrm{shot}}\)-sample multinomial histogram, while the infinite-shot curve is deterministic.
In this run, reducing the shot budget degrades accuracy.
At fixed shot budget, finite-shot accuracy does not track infinite-shot accuracy monotonically: the exact-probability curve improves and then stays high, whereas the finite-shot curves peak earlier and can decline under continued infinite-shot training.
These diagnostics identify a possible train--deploy mismatch even in this noiseless setting: exact-probability optimization can find a solution that is more fragile under a fixed inference-time shot budget.
We probe this behavior below through shot-noise propagation in the hybrid readout.

\begin{figure}
    \centering
    \includegraphics[width=\columnwidth]{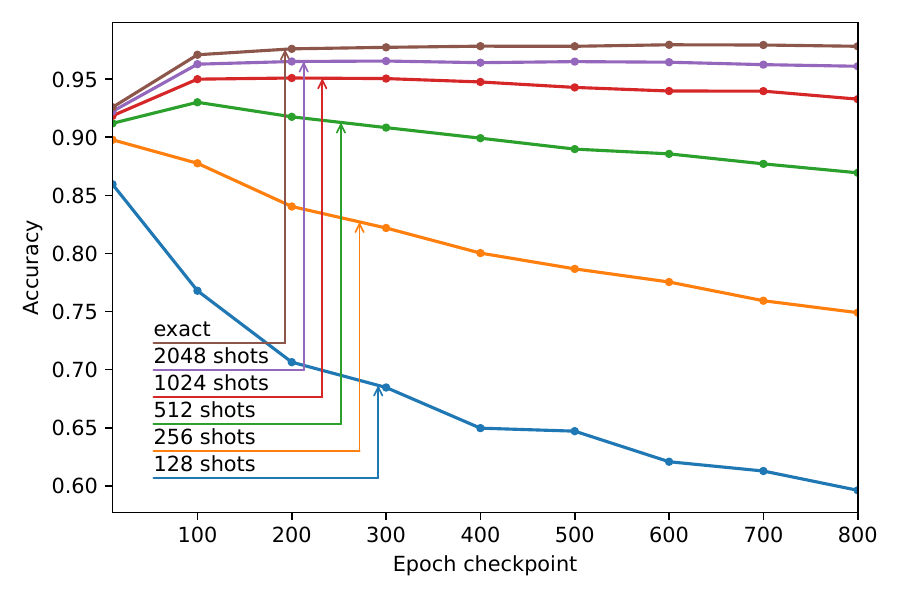}
    \caption{Finite-shot reevaluation of one representative direct-$16\times16$ full-MNIST reference PCS-QCNN ($Q=1$, $n_f=3$). Each curve corresponds to an inference-time shot budget ($128$, $256$, $512$, $1024$, $2048$, or infinite-shot exact readout). Each finite-shot point is one stochastic full-test-set pass with one independently sampled multinomial readout histogram per test image; the exact curve is deterministic. Training always uses exact probabilities; only inference is resampled.}
    \label{fig:accuracy_with_shots}
\end{figure}

\subsection{Mechanism analysis: loss distributions under finite-shot sampling}\label{sec:results_hist}

To probe this finite-shot sensitivity, we examined the distribution of loss values obtained when replacing the exact readout vector by its finite-shot estimate.
Figure~\ref{fig:hist} shows the resulting batch-mean test cross-entropy distributions for the representative direct-$16\times16$ $Q=1$, $n_f=3$ model at two training stages (epochs $100$ and $800$) and for several shot budgets.
These histograms use the same trained model.
For each epoch and shot budget, the $10{,}000$-sample test set is split into $40$ nonoverlapping batches of $250$ images; each finite-shot histogram entry is one batch-mean cross-entropy under one of $100$ independent resampling passes, giving $40\times100=4000$ entries per shot budget.
The infinite-shot reference contains one exact pass over the same $40$ batches.

\begin{figure*}[t]
    \begin{overpic}[width=0.49\textwidth]{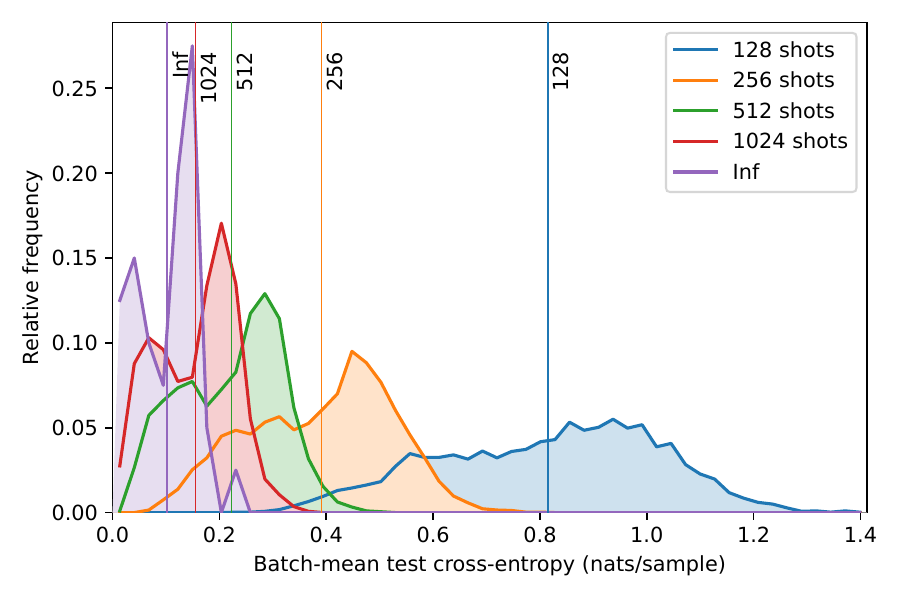}
    \put(3mm,3mm){(a)}
    \end{overpic}
    \hfill
    \begin{overpic}[width=0.49\textwidth]{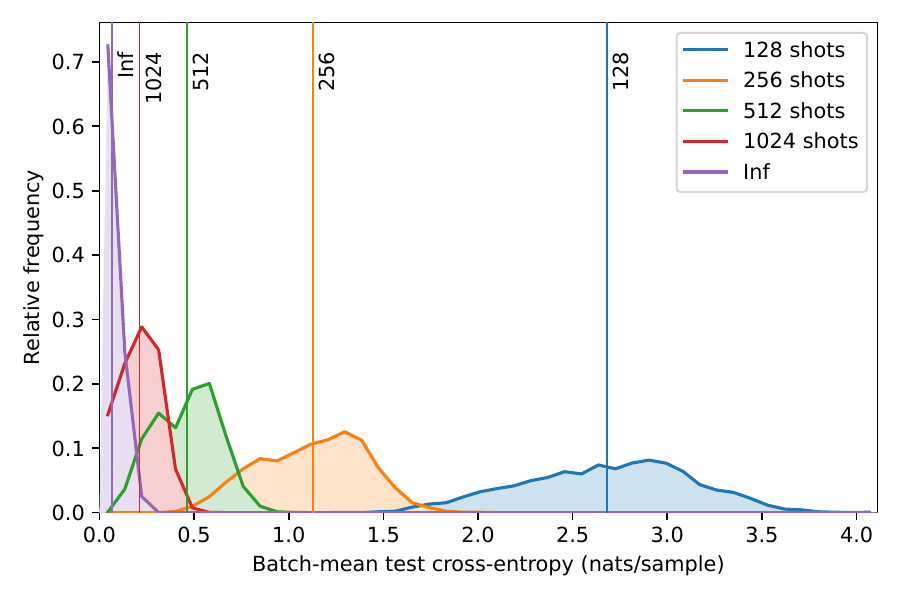}
    \put(3mm,3mm){(b)}
    \end{overpic}
    \caption{Batch-mean test cross-entropy distributions for one representative direct-$16\times16$ $Q=1$, $n_f=3$ model under finite-shot readout. (a) After $100$ training epochs. (b) After $800$ training epochs. Colored histograms correspond to $128$, $256$, $512$, $1024$, and infinite-shot readout. Each finite-shot entry is the mean cross-entropy on one $250$-sample test batch under one of $100$ independent multinomial resampling passes, giving $4000$ entries per finite-shot histogram; the exact reference has $40$ entries. Vertical lines mark sample-weighted mean losses.}
    \label{fig:hist}
\end{figure*}

Two qualitative effects are visible.
First, there exists a practical threshold in the shot budget at which the loss distribution changes character: at sufficiently large $N_{\mathrm{shot}}$ the distribution remains concentrated near the infinite-shot value, while at smaller budgets it develops pronounced tails and, in some cases, multi-modality.
Second, in this run the shape of the distribution changes with training duration: the longer-trained model can exhibit rare but extremely large loss outliers under limited shots (rightmost bins in Fig.~\ref{fig:hist}), even when the infinite-shot accuracy is higher.
These outliers are consistent with the longer-trained model being sharper with respect to perturbations of the readout distribution, so that a small probability-estimation error can occasionally push the classical head into a high-loss regime.

\subsection{Geometric view of finite-shot sensitivity}\label{sec:results_landscape}

A local loss-landscape probe in readout-probability space tests the same interpretation for the same trained model.
The resulting landscape is consistent with the same pattern: the longer-trained model has a sharper effective loss surface along finite-shot perturbation directions.
The construction and the corresponding figure are reported in Supplemental Sec.~\ref{_sec:supp_results_landscape}.

\section{Conclusions}\label{sec:conclusions}

Quantum convolution is fundamentally an encoding-aware symmetry constraint.
The relevant symmetry is not fixed by the QCNN layout alone, but by how translations of the original data act after encoding.
The circuit should be equivariant under the translation action induced on the quantum state by the chosen data encoding.
For FRQI-like address/amplitude encodings, pixel translations act as modular addition on the index registers (PCS), which generally differs from cyclic permutations of physical qubits (QCS) enforced by many MERA-inspired QCNN layouts.

This perspective yields a constructive characterization of PCS-equivariant quantum layers.
At the layer level, the resulting unitary family is the commutant of the encoded translation operator, giving a Fourier-space constructive form for PCS-equivariant quantum convolution.
We then built a multiscale PCS-QCNN quantum core with measurement-conditioned pooling and analyzed its trainability.
The trainability analysis separates the elementary parameter-count mechanism that suppresses the expected mean per-coordinate second moment from the empirical gradient norms directly relevant to optimization.
This distinction matters because a per-sample RMS diagnostic does not include cancellations between examples, whereas training follows the gradient of the averaged empirical loss.

Empirically, the benchmarks show two signatures of encoding-aligned convolution under the controls used here.
On translated-MNIST, the classical CNN/MLP control separates by $48.75$ percentage points, showing that the task is sensitive to translation-aware inductive bias, and within the quantum family the PCS-QCNN exceeds the matched non-PCS RBC-QCNN by $35.08$ percentage points.
On the full-MNIST size sweep, all larger spatial resolutions outperform the smallest $8\times8$ configuration, with direct $32\times32$ preprocessing giving the highest final mean in the displayed sweep.
The finite-shot diagnostics further show that shot budget can change the effective ranking of checkpoints, making sampling cost a deployment-relevant hyperparameter.

Taken together, these results support preserving the encoded PCS structure relative to a parameter-matched symmetry-broken quantum control in an idealized statevector benchmark.
They do not settle PCS versus QCS empirically, hardware-native performance, or systematic finite-shot degradation across independently trained models.
Those questions require matched QCS-equivariant baselines, compiled state-preparation and multiplexer circuits, structured/compressed PCS parameterizations, and shot/noise-aware training protocols.

Another open question is how far the architecture should be pushed toward an end-to-end quantum model.
In the present hybrid design, the classical head provides inexpensive nonlinear decision layers, while a fully quantum replacement would need measurement-conditioned branching or other effectively non-unitary mechanisms and could significantly increase shot requirements.
Whether such a fully quantum readout is advantageous under realistic shot budgets remains unresolved.

A second open challenge is hardware-native training.
This study uses infinite-shot simulation for optimization and finite-shot sampling at inference; extending optimization itself to shot-limited, noisy hardware requires shot-efficient gradient estimation and update rules that stay stable in large multiplexed circuits.
Related design questions include principled feature-growth policies across pooling levels and constraints on mode-dependent blocks that preserve symmetry while controlling depth and sampling cost.

\section*{Data and Code Availability}
The code that supports the findings of this study is openly available in \cite{pcsqcnn_repo}.

\begin{acknowledgments}
The research was supported by ITMO University Research Projects in AI Initiative (project 640111).
\end{acknowledgments}

\bibliography{apssamp}

{   
    \widetext
    \clearpage
    \appendix
    \begin{center}
\textbf{\large Supplemental Material for ``Pixel-Translation-Equivariant Quantum Convolutional Neural Networks via Fourier Multiplexers''}
\end{center}
\makeatletter
\@removefromreset{lemma}{section}
\@removefromreset{theorem}{section}
\@removefromreset{definition}{section}
\@removefromreset{corollary}{section}
\@removefromreset{remark}{section}
\@removefromreset{equation}{section}
\makeatother

\setcounter{lemma}{0}
\renewcommand{\thelemma}{S\arabic{lemma}}
\setcounter{theorem}{0}
\renewcommand{\thetheorem}{S\arabic{theorem}}
\setcounter{definition}{0}
\renewcommand{\thedefinition}{S\arabic{definition}}
\setcounter{corollary}{0}
\renewcommand{\thecorollary}{S\arabic{corollary}}
\setcounter{remark}{0}
\renewcommand{\theremark}{S\arabic{remark}}

\setcounter{equation}{0}
\renewcommand{\theequation}{S\arabic{equation}}
\setcounter{figure}{0}
\renewcommand{\thefigure}{S\arabic{figure}}
\setcounter{table}{0}
\renewcommand{\thetable}{S\arabic{table}}

This document includes
a literature table for QCNN-style MNIST results (Sec.~\ref{sec:supp_mnist_context}),
additional experiments and analyses (Sec.~\ref{sec:supp_additional_experiments}),
an explicit specification of the classical baseline architectures used in benchmark comparisons (Sec.~\ref{sec:supp_classical_baselines}),
the full-MNIST classical control used to motivate the translated benchmark choice (Sec.~\ref{sec:supp_full_mnist_classical_baselines}),
an extended classical convolution primer (Sec.~\ref{sec:supp_classical_primer}),
complete proofs of the QCS--PCS mismatch lemma and the Fourier-multiplexer theorem (Sec.~\ref{sec:supp_symmetry_proofs}),
a longer discussion of MERA-QCNN-style (QCS-equivariant) templates (Sec.~\ref{sec:supp_mera_templates}),
derivations behind the Fourier-junction reduction between pooled PCS layers (Sec.~\ref{sec:supp_fourier_junction}),
discussion of the idealized benchmark model and its resource scaling (Sec.~\ref{sec:supp_resource_scaling}),
and the gradient-accounting derivation (Sec.~\ref{sec:supp_trainability_full}).

\section{QCNN-style models on MNIST: literature table}\label{sec:supp_mnist_context}

MNIST~\cite{LeCunMNIST} is widely used in demonstrations of quantum and hybrid quantum--classical image classifiers, including QCNN-inspired architectures.
However, the literature is highly heterogeneous in at least four respects: (i) binary vs multi-class tasks (often with different label subsets), (ii) image resolution and preprocessing (downscaling, PCA-based compression, handcrafted features), (iii) the training/evaluation protocol (data splits, subset selection, number of epochs), and (iv) the readout model (infinite-shot/statevector versus finite-shot sampling, and sometimes implicit assumptions about state preparation).
As a result, reported accuracies are rarely directly comparable across papers.
This is a general issue in quantum-machine-learning benchmarking, where design choices and classical controls can affect conclusions~\cite{Bowles2024}.

In the main text (Sec.~\ref{sec:mnist_motivation}), MNIST is used in two fixed benchmark families: a translated-MNIST regime where the classical convolutional and dense baselines are well separated, and a full-MNIST size sweep used for the main infinite-shot and finite-shot PCS-QCNN analyses (Secs.~\ref{sec:results_infinite_shot}--\ref{sec:results_hist}).

Table~\ref{tab:mnist_lit} summarizes representative QCNN-like results on MNIST, together with the corresponding task type and input resolution as reported by the authors.

\begin{table*}[ht]
\centering
\begin{threeparttable}
\small
\setlength{\tabcolsep}{5pt}
\begin{tabular}{p{0.24\textwidth}p{0.13\textwidth}p{0.29\textwidth}p{0.26\textwidth}}
\toprule
\textbf{Work} & \textbf{Task} & \textbf{Input / preprocessing} & \textbf{Reported accuracy (\%)} \\
\midrule
Mattern et al.~\cite{Mattern2021} & 10-class & $14\times14$ & 85--88 \\
Huang et al.~\cite{Huang2023} & 10-class & $14\times14$ & 96 \\
Li et al.~\cite{Li2022} & Binary & $4\times4$ (label pairs) & 99.80 (4 vs 5), 91.51 (3 vs 8) \\
Oh et al.~\cite{Oh2021} & 10-class & $10\times10$ & 95 \\
Huggins et al.~\cite{Huggins2018} & Binary & $8\times8$ & 95 \\
Jing et al.~\cite{Jing2022} & 10-class & $10\times10$ and $20\times20$ & 94--95 \\
Easom-McCaldin et al.~\cite{Easom_Mccaldin} & Binary & $9\times9$ & 100 \\
Zheng et al.~\cite{Zheng2023} & Binary & $8\times8$ & 96.65 \\
Hur et al.~\cite{Hur2022} & Binary & PCA to $16\times16$ & 98.1 \\
Huang et al.~\cite{Huang2022} & Binary / 3-class & $8\times8$ & 95.8 / 72.1 \\
Li et al.~\cite{Li2020} & 10-class & $32\times32$ & 98.97 \\
Wei et al.~\cite{Wei2022} & 10-class & $32\times32$ & 96.3 and 74.3 \\
\midrule
\textbf{This work} & 10-class & full-MNIST, direct $16\times16$ / direct $32\times32$ preprocessing & 97.96 / 98.13 (both exact-probability inference) \\
\bottomrule
\end{tabular}
\begin{tablenotes}[flushleft]
\footnotesize
\item Notes: Accuracy values are reproduced from the cited papers, higher is better.
\end{tablenotes}
\end{threeparttable}
\caption{Overview of representative QCNN-like MNIST results in the literature.
Entries are shown in the form reported by the corresponding references.
Our main paper uses MNIST with a fixed preprocessing and evaluation protocol and includes within-family matched controls (Sec.~\ref{sec:results_baselines}).}
\label{tab:mnist_lit}
\end{table*}

\section{Additional experiments and analyses}\label{sec:supp_additional_experiments}

This section gives additional analyses for the text Results section (Sec.~\ref{sec:results}): an encoder-scale brightness sensitivity sweep, an extended hyperparameter and parameter-accounting report, a readout entropy diagnostic, a local readout-space loss-landscape probe, and qualitative error analysis.
These analyses document the sensitivity of the FRQI angle map, training behavior, finite-shot sensitivity, and residual error structure.

\subsection{Hyperparameter sweep and parameter accounting}\label{sec:supp_results_sweep}

The experiments reported here use two sweep families.
Figure~\ref{fig:diff_models}(a) is a translated-MNIST architecture sweep with digits resized to $16\times16$, placed on a $32\times32$ canvas, translated by at most $8$ pixels along each axis, evaluated with full readout, and trained for $2000$ epochs over $3$ seeds, with $Q\in\{1,2,3,4,5\}$ and $n_f\in\{1,2,3\}$.
Figure~\ref{fig:diff_models}(b) is a full-MNIST size sweep with full readout, $Q=1$, $n_f=3$, $2000$ training epochs, and $3$ seeds, across four preprocessing settings: direct $8\times8$, direct $16\times16$, centered $28\times28$ in a $32\times32$ canvas, and direct $32\times32$.
The follow-up reference model reused in Figs.~\ref{fig:accuracy_with_shots}, \ref{fig:hist}, and Supplemental Figs.~\ref{fig:readout_entropy}, \ref{fig:ellipse}, and~\ref{fig:decomp} is the medium-resolution direct $16\times16$ member of this full-MNIST size sweep.

For square canvas size $2^{n_{\mathrm{idx}}}$ and full readout, the total logical qubit count is \(n_{\mathrm{tot}}=2n_{\mathrm{idx}}+n_f\) (Eq.~\eqref{eq:q_tot}).
After $Q-1$ pooling steps per axis, the remaining active index width is \(n_l=n_{\mathrm{idx}}-Q+1\), so the classifier sees \(n_{\mathrm{meas}}=2n_l+n_f\) effective measured qubits (Eq.~\eqref{eq:q_meas}) and input dimension \(D_{\mathrm{out}}=2^{n_{\mathrm{meas}}}\) (Eq.~\eqref{eq:readout_dim}).
Because the classifier is a single biased linear layer, its parameter count is exactly $10D_{\mathrm{out}}+10$.
The quantum parameter counts follow directly from the model definition in the main text.
Tables~\ref{tab:results_architecture} and \ref{tab:results_size} summarize these exact counts and final mean accuracies for the reported sweeps; the full accuracy trajectories are shown in Figs.~\ref{fig:diff_models}(a,b).

The translated-MNIST architecture sweep shows that intermediate depths and wider feature registers can improve performance, with the largest displayed mean at \(Q=3,n_f=3\); with only three seeds and no uncertainty column in the table, this should be read as a descriptive sweep result rather than a resolved architecture ranking.
The full-MNIST size sweep likewise depends on image resolution: direct $8\times8$ gives the lowest final mean ($93.21\%$), the larger three preprocessing choices span $97.96$--$98.13\%$, and the direct $32\times32$ model attains the highest final mean in the displayed sweep ($98.13\%$).

\begin{table*}[ht]
    \centering
    \begin{threeparttable}
    \small
    \setlength{\tabcolsep}{5pt}
    \begin{tabular}{@{}lcrrll@{}}
    \toprule
    \textbf{Configuration} & \textbf{Total qubits} & \textbf{Quantum params} & \textbf{Classifier params} & \textbf{Readout shape} & \textbf{Final mean acc. (\%)} \\
    \midrule
    $Q=1$, $n_f=1$ & 11 & 4\,096 & 20\,490 & $32\times32\times2$ & 51.29 \\
    $Q=1$, $n_f=2$ & 12 & 16\,384 & 40\,970 & $32\times32\times4$ & 62.50 \\
    $Q=1$, $n_f=3$ & 13 & 65\,536 & 81\,930 & $32\times32\times8$ & 66.36 \\
    $Q=2$, $n_f=1$ & 11 & 8\,192 & 5\,130 & $16\times16\times2$ & 59.02 \\
    $Q=2$, $n_f=2$ & 12 & 32\,768 & 10\,250 & $16\times16\times4$ & 76.00 \\
    $Q=2$, $n_f=3$ & 13 & 131\,072 & 20\,490 & $16\times16\times8$ & 76.85 \\
    $Q=3$, $n_f=1$ & 11 & 9\,216 & 1\,290 & $8\times8\times2$ & 65.02 \\
    $Q=3$, $n_f=2$ & 12 & 36\,864 & 2\,570 & $8\times8\times4$ & 75.89 \\
    $Q=3$, $n_f=3$ & 13 & 147\,456 & 5\,130 & $8\times8\times8$ & 79.15 \\
    $Q=4$, $n_f=1$ & 11 & 9\,472 & 330 & $4\times4\times2$ & 62.63 \\
    $Q=4$, $n_f=2$ & 12 & 37\,888 & 650 & $4\times4\times4$ & 74.44 \\
    $Q=4$, $n_f=3$ & 13 & 151\,552 & 1\,290 & $4\times4\times8$ & 73.79 \\
    $Q=5$, $n_f=1$ & 11 & 9\,536 & 90 & $2\times2\times2$ & 60.46 \\
    $Q=5$, $n_f=2$ & 12 & 38\,144 & 170 & $2\times2\times4$ & 70.25 \\
    $Q=5$, $n_f=3$ & 13 & 152\,576 & 330 & $2\times2\times8$ & 67.96 \\
    \bottomrule
    \end{tabular}
    \begin{tablenotes}[flushleft]
    \footnotesize
    \item Notes: This table corresponds to the translated-MNIST architecture sweep in Fig.~\ref{fig:diff_models}(a). All runs use a $32\times32$ canvas with digits resized to $16\times16$ and translated by at most $8$ pixels per axis. The accuracy column reports the final mean test accuracy over $3$ seeds from the runs used for the figure; small differences among high-performing rows should not be interpreted as a statistically resolved ranking.
    \end{tablenotes}
    \end{threeparttable}
\caption{Exact parameter accounting for the PCS-QCNN architecture sweep in Fig.~\ref{fig:diff_models}(a).}
    \label{tab:results_architecture}
\end{table*}

\begin{table*}[t]
    \centering
    \begin{threeparttable}
    \small
    \setlength{\tabcolsep}{5pt}
    \begin{tabular}{@{}lcrrll@{}}
    \toprule
    \textbf{Preprocessing} & \textbf{Total qubits} & \textbf{Quantum params} & \textbf{Classifier params} & \textbf{Readout shape} & \textbf{Final mean acc. (\%)} \\
    \midrule
    direct $8\times8$ & 9 & 4\,096 & 5\,130 & $8\times8\times8$ & 93.21 \\
    direct $16\times16$ & 11 & 16\,384 & 20\,490 & $16\times16\times8$ & 97.96 \\
    $28\times28$ on $32\times32$ & 13 & 65\,536 & 81\,930 & $32\times32\times8$ & 97.96 \\
    direct $32\times32$ & 13 & 65\,536 & 81\,930 & $32\times32\times8$ & 98.13 \\
    \bottomrule
    \end{tabular}
    \begin{tablenotes}[flushleft]
    \footnotesize
    \item Notes: This table corresponds to the full-MNIST size sweep in Fig.~\ref{fig:diff_models}(b), which fixes $Q=1$ and $n_f=3$ and varies only the preprocessing choice. The $28\times28$-on-$32\times32$ and direct-$32\times32$ rows share the same quantum and classifier dimensions because both use a $32\times32$ quantum canvas; they differ only in image preprocessing before encoding, with no random translations in either case. The accuracy column reports the final mean test accuracy over $3$ seeds from the runs used for the figure.
    \end{tablenotes}
    \end{threeparttable}
\caption{Exact parameter accounting for the PCS-QCNN size sweep in Fig.~\ref{fig:diff_models}(b).}
    \label{tab:results_size}
\end{table*}

\subsubsection{Preprocessing convention for encoder inputs.}\label{sec:encoder_preproc}
For reproducibility, let $x_{u,v}\in[0,1]$ denote the preprocessed grayscale intensity at pixel $(u,v)$ after resizing, optional canvas placement, and normalization.
The encoder maps it to an angle
\[
    p_{u,v}=a+(b-a)x_{u,v}.
\]
For the PCS-QCNN experiments reported in the paper, $(a,b)=(0,\pi)$.
Thus $x_{u,v}$ is the direct input to the brightness map in Eq.~\eqref{eq:color_map_impl}.
For \(n_f>1\), the auxiliary feature qubits are initialized in \(\ket{0}^{\otimes(n_f-1)}\), so the local feature-register state is \(\ket{0}^{\otimes(n_f-1)}\otimes\ket{\phi_{u,v}}\) with the grayscale/color state on the least-significant feature qubit.
In the numerical protocol used here, the global $1/\sqrt{N_xN_y}$ prefactor from Eq.~\eqref{eq:frqi_like} is omitted during state initialization; measurement compensates for this by dividing the final marginal by the corresponding overall spatial normalization factor.

\subsubsection{Block parameterization in the benchmark model.}
In the benchmark model, each mode-wise feature-register block $U^{(\ell)}_{k_x,k_y}(m)$ is parameterized as a general $U(2^{n_f})$ unitary via the full Pauli-basis exponential (Eq.~\eqref{eq:pauli_exp_param} in the main text, Sec.~\ref{sec:benchmark_models}).
Concretely, this includes the all-identity Pauli string and yields $4^{n_f}$ real parameters per Fourier mode (and, for $\ell\ge 2$, per selected condition branch), and the full multiplexer $\mathcal{B}^{(\ell)}$ is applied as an explicit block-diagonal unitary in the Fourier basis.
The identity coefficient is retained because block-local phases become relative phases between Fourier modes or pooling branches in the full multiplexer.
For simulation, we treat this multiplexer as a single explicit block-diagonal unitary map.
Hardware compilation cost is discussed in Sec.~\ref{sec:supp_complexity}.

\subsubsection{Numerical realization.}
The reported training and evaluation use dense-tensor statevector simulation rather than a gate-by-gate hardware emulation.
The PCS-QCNN evolution, including Fourier transforms, multiplexers, marginal measurement, and finite-shot sampling, is evaluated directly on the full state tensor.
This choice isolates the architectural questions studied in the paper from compilation overhead and hardware noise.

\subsection{Brightness-range sensitivity sweep for encoder scaling}\label{sec:brightness_sweep}

Besides the main translated-MNIST architecture sweep and the full-MNIST size sweep, we ran a separate low-data sensitivity sweep for the encoder angle scale.
The article-wide convention $(a,b)=(0,\pi)$ was chosen a priori as a one-period real FRQI color-map convention and was not selected by optimizing this sweep.
The sweep is included as a post-hoc sanity check of sensitivity to the upper endpoint \(b\), not as a validation procedure for tuning the reported benchmark models.
No reported benchmark setting, architecture choice, or quoted result accuracy is selected using the accuracy values in this sweep.

The sweep uses translated-MNIST with $20$ training examples per class, while keeping the standard test split membership unchanged.
Both train and test images are resized to $28\times28$, placed on a $32\times32$ canvas, and translated by independently sampled integer offsets of at most $2$ pixels per axis.
For each independent run, the prepared images are fixed and reused for all epochs and evaluations.
The PCS-QCNN is evaluated with full readout for $300$ training epochs, train/test batch sizes $256/1600$, and $3$ seeds.
We evaluate PCS-QCNN architectures with feature-qubit counts $n_f\in\{1,2\}$ and quantum depths $Q\in\{1,2\}$.
For each architecture we keep the lower endpoint fixed at $a=0$ and vary the upper endpoint as $b=\beta\pi$ with
\[
\beta\in\left\{\frac{2}{25},\frac{4}{25},\dots,\frac{48}{25}\right\},
\]
that is, $24$ evenly spaced interior points spanning the open interval $(0,2\pi)$.
Accuracy is recorded at epochs $150$ and $300$, corresponding to the two panels in Fig.~\ref{fig:brightness_sweep}.
For each architecture and each value of $\beta$, the line shows the mean over seeds and the band shows the seedwise minimum--maximum range.

\begin{figure*}[t]
    \hfill
    \begin{overpic}[width=0.45\textwidth]{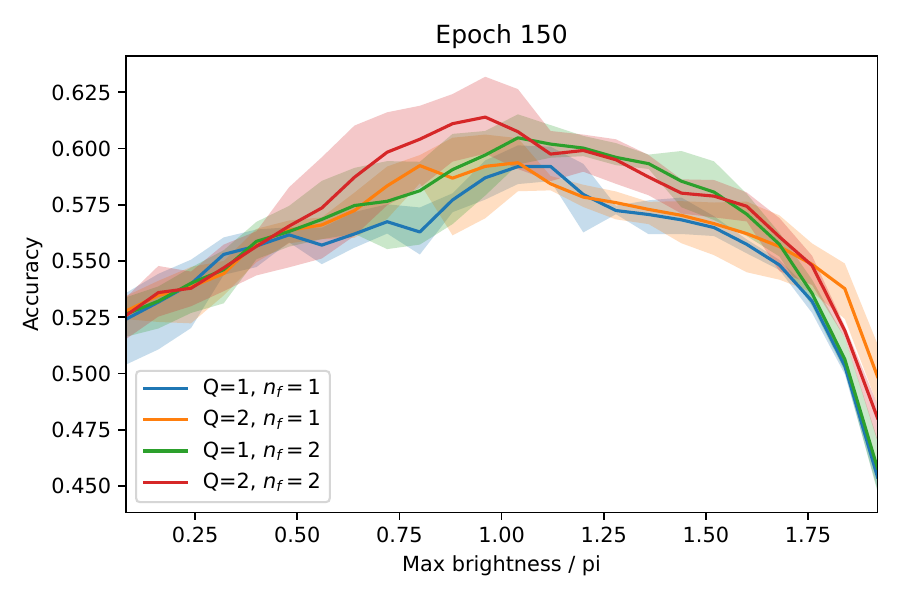}
    \put(0mm,0mm){(a)}
    \end{overpic}
    \hfill
    \begin{overpic}[width=0.45\textwidth]{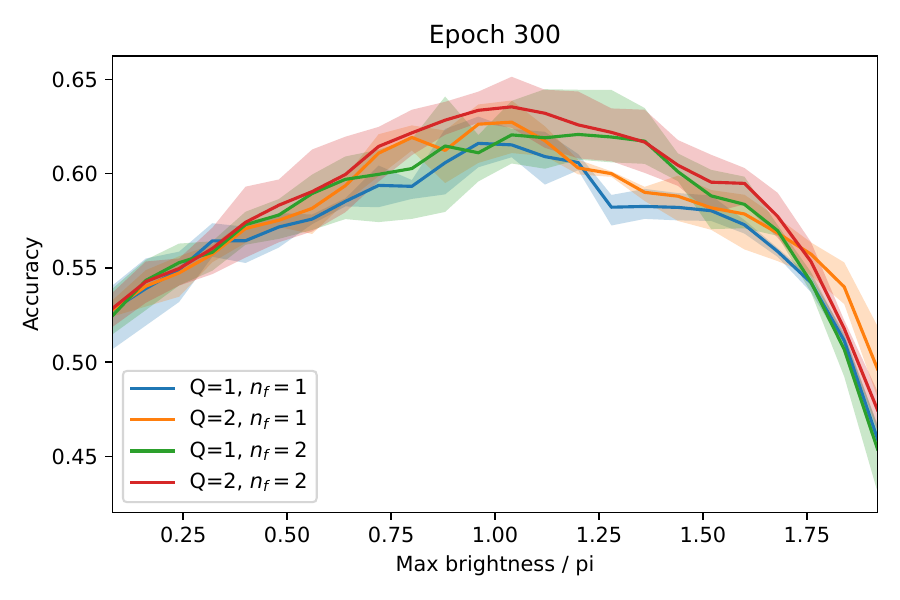}
    \put(0mm,0mm){(b)}
    \end{overpic}
    \hfill
    \caption{Brightness-range sensitivity sweep for the FRQI encoder scale. Panels show accuracy at epochs $150$ and $300$ as a function of the upper endpoint $b/\pi$ in the map $p=a+(b-a)x$ with fixed lower endpoint $a=0$. The article-wide benchmark convention $(a,b)=(0,\pi)$ was fixed independently of this sweep. Each curve corresponds to one architecture $(Q,n_f)\in\{1,2\}\times\{1,2\}$; lines show the mean over $3$ seeds and shaded bands show the seedwise minimum--maximum range.}
    \label{fig:brightness_sweep}
\end{figure*}

The sweep is a post-hoc sensitivity check around the fixed $[0,\pi]$ one-period convention for the local real FRQI color map; it does not imply that longer intervals are strictly redundant for the full coherent image state.
With
\[
\ket{\phi(p)}=\sin(p)\ket{0}+\cos(p)\ket{1},
\]
one has $\ket{\phi(p+\pi)}=-\ket{\phi(p)}$, so the isolated one-qubit ray is revisited after a shift by $\pi$.
In the FRQI image superposition, however, address-dependent sign changes are relative phases across address components and can affect subsequent QFT/PCS layers.
The brightness sweep is therefore reported as a robustness/sensitivity diagnostic only.
All reported benchmark families fix the encoder interval to $(a,b)=(0,\pi)$ independently of this sweep.

\subsection{Classical baseline architectures for benchmark comparisons}\label{sec:supp_classical_baselines}

The classical reference models used in Fig.~\ref{fig:small_ds_clas}(a) are specified here for reproducibility (text Secs.~\ref{sec:benchmark_baselines} and~\ref{sec:results_baselines}).

For the CNN baseline (Fig.~\ref{fig:supp_classical_baselines}(a)), the input is a $32\times32$ grayscale canvas produced by the shared preprocessing protocol.
In Fig.~\ref{fig:small_ds_clas}(a), this architecture is evaluated on the translated benchmark with $16\times16$ digits placed on a $32\times32$ canvas, and the same network is also used for the full-MNIST control without translations.
The network applies
\[
\mathrm{Conv}(1\!\to\!16,3\times3)\to\mathrm{ReLU}\to\mathrm{Conv}(16\!\to\!32,3\times3)\to\mathrm{ReLU}\to\mathrm{AvgPool}(2\times2)
\]
followed by
\[
\mathrm{Conv}(32\!\to\!48,3\times3)\to\mathrm{ReLU}\to\mathrm{Conv}(48\!\to\!64,3\times3)\to\mathrm{ReLU}\to\mathrm{AvgPool}(2\times2)
\]
and finally \(\mathrm{AdaptiveAvgPool}(1\times1)\to\mathrm{Dropout}(0.10)\to\mathrm{Linear}(64\to10)\).
This model has $47{,}034$ trainable parameters.

For the MLP baseline (Fig.~\ref{fig:supp_classical_baselines}(b)), the same $32\times32$ input is flattened to length $1024$ and passed through
\[
\mathrm{Linear}(1024\to29)\to\mathrm{GELU}\to\mathrm{Dropout}(0.10)\to\mathrm{Linear}(29\to116)\to\mathrm{GELU}\to\mathrm{Dropout}(0.10)
\]
followed by
\[
\mathrm{Linear}(116\to116)\to\mathrm{GELU}\to\mathrm{Dropout}(0.10)\to\mathrm{Linear}(116\to10).
\]
For $32\times32$ inputs this model has $47{,}947$ trainable parameters.

For comparison, the fixed translated-MNIST PCS-QCNN and the matched random-basis control in Fig.~\ref{fig:small_ds_clas}(b) both use $Q=3$, $n_f=2$, full readout on the same translated benchmark with $16\times16$ digits placed on a $32\times32$ canvas, and $39{,}434$ total trainable parameters ($36{,}864$ in the quantum core and $2{,}570$ in the classifier head).
This random-basis control is a non-PCS symmetry-breaking ablation of the Fourier construction, not an explicitly QCS-equivariant or MERA-style baseline.

\begin{figure*}[t]
    \centering
    \begin{overpic}[width=0.55\textwidth]{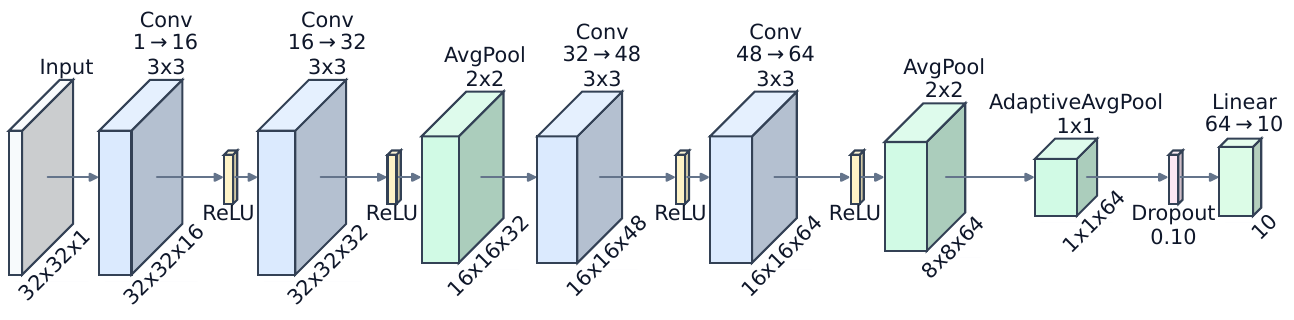}
    \put(0mm,21mm){(a)}
    \end{overpic}
    \begin{overpic}[width=0.44\textwidth]{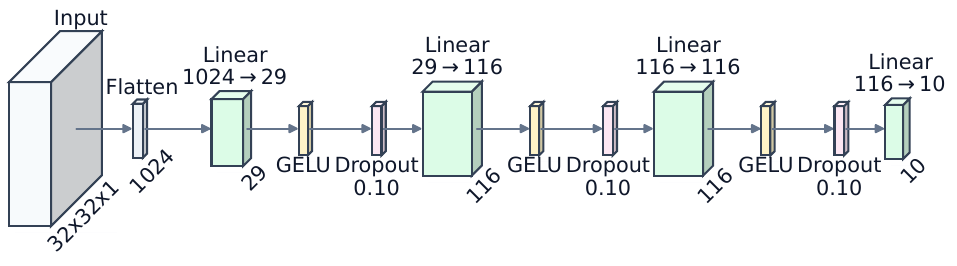}
    \put(0mm,21mm){(b)}
    \end{overpic}
    \caption{Classical reference architectures used in Fig.~\ref{fig:small_ds_clas}(a). The same $32\times32$ input models are used for the translated benchmark in the main text and for the full-MNIST control in Sec.~\ref{sec:supp_full_mnist_classical_baselines}. (a) Convolutional CNN baseline with $47{,}034$ trainable parameters. (b) Pure-dense MLP control with $47{,}947$ trainable parameters.}
    \label{fig:supp_classical_baselines}
\end{figure*}

\subsection{Full-MNIST classical control without translations}\label{sec:supp_full_mnist_classical_baselines}

The same CNN and MLP baselines were also evaluated on the full standard MNIST split ($60{,}000$ training images and $10{,}000$ test images) after resizing each digit directly to a $32\times32$ grayscale image, with no canvas translation.
The model architectures were kept identical to those specified in Sec.~\ref{sec:supp_classical_baselines}, so the comparison changes only the data regime and not the parameter budget.
Training followed the same fixed-baseline protocol as the text classical comparison except for the training horizon: three random seeds, $1000$ epochs, train/test accuracy tracked through time, and the same percentile-band summary convention as in Fig.~\ref{fig:small_ds_clas}(a).

Figure~\ref{fig:supp_full_mnist_classical_baselines} shows that on this non-translated full-MNIST task the convolutional and dense baselines remain much closer than on translated-MNIST, with final mean test accuracies of $99.15\%$ and $96.74\%$, respectively.
This control is the reason the main text uses the translated benchmark with $16\times16$ digits on a $32\times32$ canvas when discussing convolution-sensitive inductive bias: removing translations makes the classical gap much less informative.

\begin{figure*}[t]
    \centering
    \includegraphics[width=0.5\textwidth]{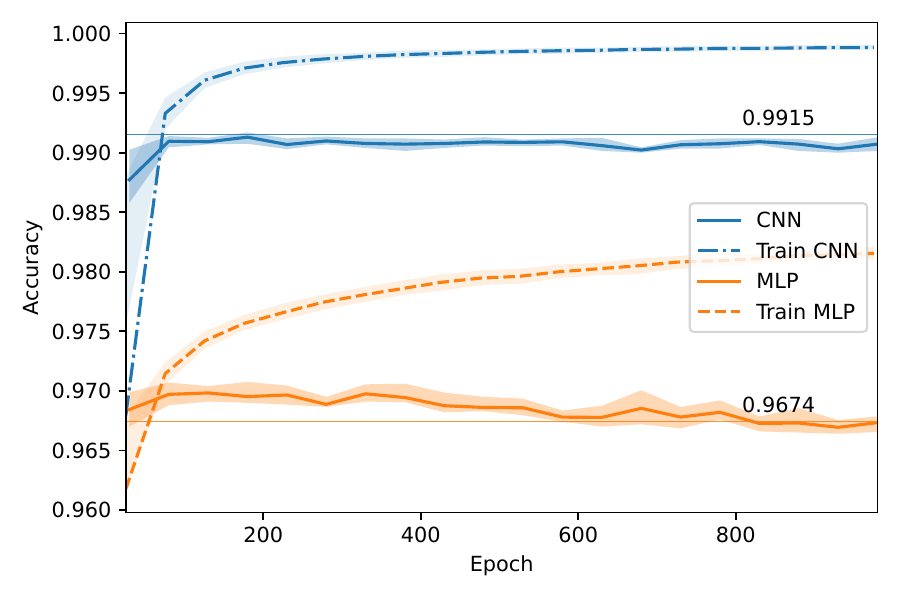}
    \caption{Classical CNN and MLP controls on the full standard MNIST split resized directly to $32\times32$ without translations. The architectures are exactly those specified in Fig.~\ref{fig:supp_classical_baselines}; only the data regime changes. Solid lines show mean test accuracy, train curves show mean train accuracy, and shaded bands show the $25$th--$75$th percentile range over $3$ seeds. The gap is much smaller than on the translated benchmark, which illustrates why non-translated MNIST alone is a less stringent benchmark for testing convolutional inductive bias.}
    \label{fig:supp_full_mnist_classical_baselines}
\end{figure*}

\subsection{Readout entropy diagnostic}\label{sec:supp_readout_entropy}

The main text reports initialization-time quantum-gradient diagnostics in Fig.~\ref{fig:gradient_norms_main}, distinguishing the norm of the empirical-loss gradient from the per-sample RMS gradient norm.
The readout entropy diagnostic below checks how finite-shot sampling changes the classifier input.
Figure~\ref{fig:readout_entropy} shows the Shannon entropy of the full readout distribution under finite-shot reevaluation for the representative direct-$16\times16$ reference run after epochs $100,200,\dots,2000$ and shot budgets $128$, $256$, $512$, $1024$, and $2048$; each curve reports the mean over test samples and the shaded band shows the interquartile range.
The lower entropy at smaller shot budgets is partly expected for a simple combinatorial reason: an $N$-shot histogram can have at most $N$ nonzero outcomes.
Within that limitation, the entropy increases moderately during training, indicating that sampled readouts spread over more outcomes instead of collapsing onto a small subset.
Thus, finite-shot readout can distort the effective classifier input even when the exact-probability model has a stable gradient signal.

\begin{figure}
    \centering
    \includegraphics[width=0.65\textwidth]{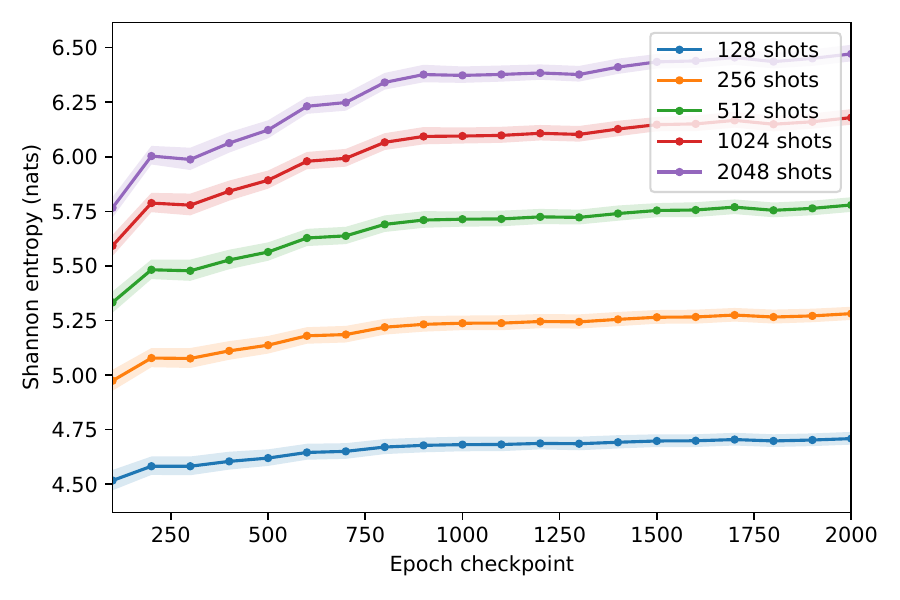}
    \caption{Mean Shannon entropy of the full readout distribution versus training epoch for the representative direct-$16\times16$ reference run after epochs $100,200,\dots,2000$ and finite-shot budgets $128$, $256$, $512$, $1024$, and $2048$; shaded bands show the interquartile range over test samples.}
    \label{fig:readout_entropy}
\end{figure}

\subsection{Geometric interpretation: local loss landscape in readout space}\label{sec:supp_results_landscape}

The loss-landscape diagnostic probes the classical head under finite-shot perturbations expected from the multinomial readout model.
Figure~\ref{fig:ellipse} uses the representative direct-$16\times16$ full-MNIST reference PCS-QCNN after $100$ and $800$ training epochs, shot budget \(N=128\), and an \(81\times81\) grid of coordinates \((\alpha,\beta)\in[-3,3]\times[-3,3]\).
For each test sample \(x\) with label \(y\), let \(p(x)\) denote the exact infinite-shot readout distribution produced by the trained quantum part, and let \(h\) denote the fixed trained classical head mapping readout vectors to class logits.
We define the sample-local multinomial covariance
\[
C_x=\frac{\diag(p(x))-p(x)p(x)^\top}{N},
\qquad N=128,
\]
take its two leading eigenpairs \((\lambda_1(x),e_1(x))\) and \((\lambda_2(x),e_2(x))\), and form the two local perturbation directions
\[
u_x=\sqrt{\lambda_1(x)}\,e_1(x),
\qquad
v_x=\sqrt{\lambda_2(x)}\,e_2(x).
\]
These directions define local sigma coordinates in readout-probability space for that sample.
At each grid point \((\alpha,\beta)\), we perturb the exact readout as
\[
q_x(\alpha,\beta)=p(x)+\alpha u_x+\beta v_x.
\]
No renormalization is applied.
A sample contributes to a grid cell only when every component of \(q_x(\alpha,\beta)\) is nonnegative and \(\sum_j q_{x,j}(\alpha,\beta)\le 1+10^{-5}\).
Writing \(\mathcal{V}_{\alpha,\beta}\subseteq\mathcal{T}\) for the subset of test samples satisfying this validity condition at grid point \((\alpha,\beta)\), the plotted quantities are
\[
L(\alpha,\beta)=\frac{1}{|\mathcal{V}_{\alpha,\beta}|}
\sum_{(x,y)\in\mathcal{V}_{\alpha,\beta}}
\ell\!\left(h\!\left(q_x(\alpha,\beta)\right),y\right),
\qquad
f_{\mathrm{valid}}(\alpha,\beta)=\frac{|\mathcal{V}_{\alpha,\beta}|}{|\mathcal{T}|},
\]
where \(\ell\) is the usual cross-entropy and \(\mathcal{T}\) is the full test set.
At \((\alpha,\beta)=(0,0)\), no perturbation is applied, so \(q_x(0,0)=p(x)\) and the heatmap value equals the exact infinite-shot mean test loss of that trained model.
White cells indicate grid points with \(f_{\mathrm{valid}}(\alpha,\beta)<0.10\), i.e.\ fewer than \(10\%\) of test samples yield valid perturbed readouts there.

\begin{figure*}[t]
    \centering
    \begin{overpic}[width=0.49\textwidth]{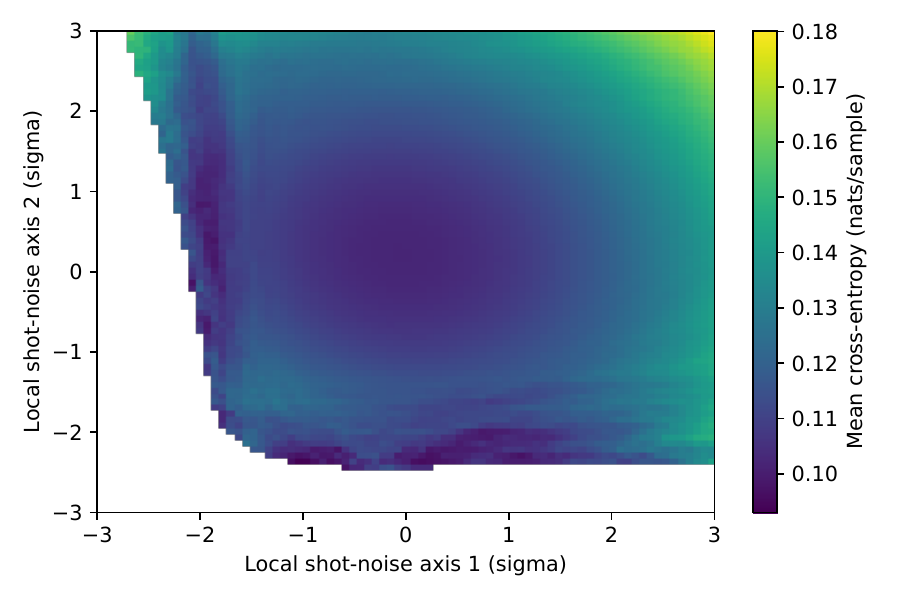}
    \put(3mm,3mm){(a)}
    \end{overpic}
    \hfill
    \begin{overpic}[width=0.49\textwidth]{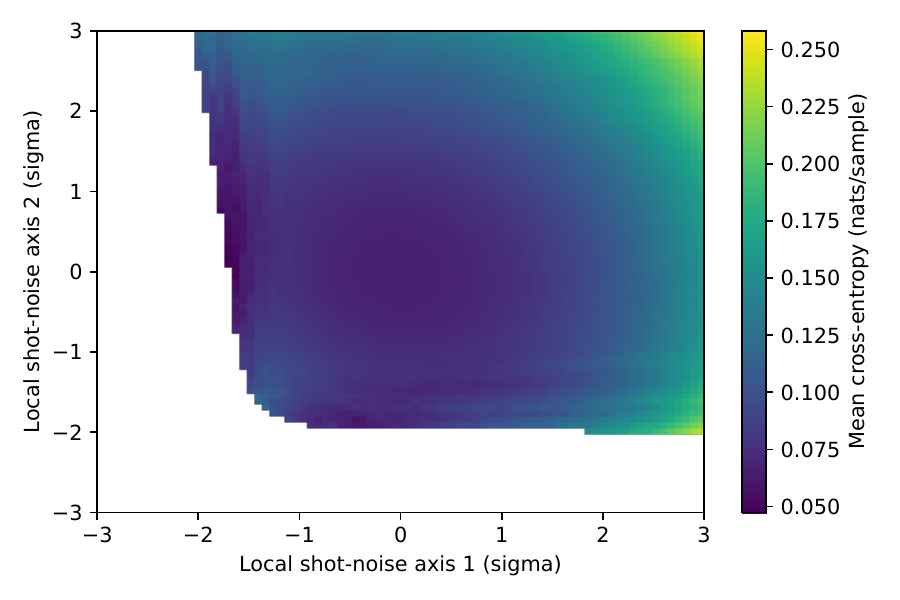}
    \put(3mm,3mm){(b)}
    \end{overpic}
    \caption{Readout-space loss landscape for the representative direct-$16\times16$ reference PCS-QCNN after (a) $100$ and (b) $800$ training epochs. Each test sample is perturbed in its own two-dimensional local shot-noise PCA basis, the resulting cross-entropy is averaged over valid test samples on an \(81\times81\) grid covering \([-3,3]^2\) in sigma units for \(N=128\) shots, and cells with valid fraction below \(0.10\) are shown in white. Each panel has its own colorbar.}
    \label{fig:ellipse}
\end{figure*}

For the model trained for $100$ epochs (Fig.~\ref{fig:ellipse}(a)), the loss rises relatively gradually away from the exact solution over most of the region that remains well defined.
After $800$ epochs (Fig.~\ref{fig:ellipse}(b)), the center value is lower, reflecting the lower exact infinite-shot test loss, but the local increase away from the center is steeper and higher losses are reached within only a few shot-noise standard deviations.
Because each panel is rendered with its own colorbar, this comparison should be based on the numeric scales and the local rise away from the center rather than on hue alone.
In this empirical sense, the $800$-epoch model is locally sharper in readout space along typical finite-shot perturbation directions than the $100$-epoch model.
This local sharpening is consistent with the observed need, in this run, for shot budgets on the order of \(10^3\) to reliably preserve the gains of longer exact-probability training.

\subsection{Qualitative error structure}\label{sec:supp_results_errors}

Figure~\ref{fig:decomp} reports the remaining classification errors for the representative final direct-$16\times16$ full-MNIST reference model after $2000$ training epochs (confusion matrix and examples of misclassified digits).
These errors are not dominated by a single failure mode; rather, they reflect the expected ambiguity between visually similar digits after aggressive downscaling to a $16\times16$ quantum input.

\begin{figure*}[t]
    \centering
    \begin{overpic}[width=0.49\textwidth]{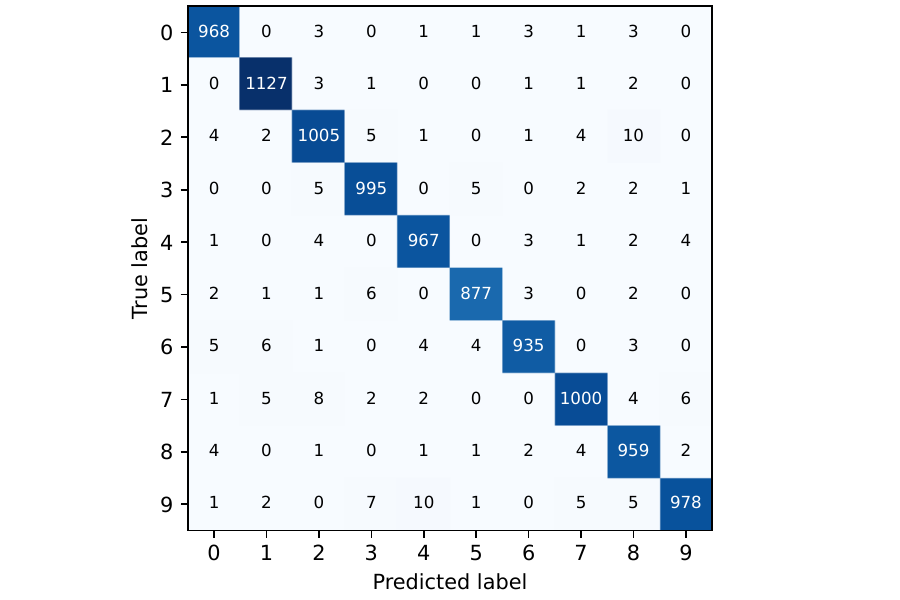}
    \put(3mm,5mm){(a)}
    \end{overpic}
    \begin{overpic}[width=0.49\textwidth]{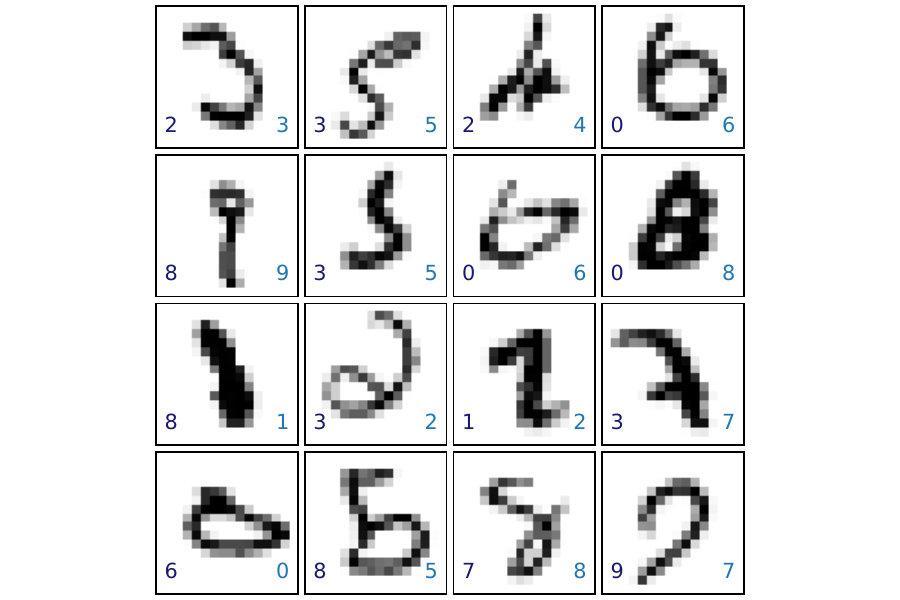}
    \put(0mm,5mm){(b)}
    \end{overpic}
    \caption{Qualitative error analysis for the representative final direct-$16\times16$ full-MNIST PCS-QCNN after $2000$ training epochs. (a) Confusion matrix (rows: true labels, columns: predicted labels); darker cells indicate higher counts.
    (b) Examples of misclassified MNIST digits with predicted labels shown in the lower-left corner and true labels in the lower-right corner of each image.}
    \label{fig:decomp}
\end{figure*}

\section{Classical convolution primer}\label{sec:supp_classical_primer}

The main text (Sec.~\ref{sec:classical_convolution}) states the key symmetry fact (convolution as commutation with translations) and the Fourier diagonalization template~\cite{bamieh2022discoveringtransformstutorialcirculant,Gray2006}.
A slightly more detailed classical discussion is reproduced here.

Beyond this algebraic characterization, two recent theory papers are directly relevant to our benchmarking logic.
Li \emph{et al.}~\cite{liWhyAreConvolutional2021} exhibit settings where convolutional architectures can be learned with fewer samples than parameter-matched fully connected models under standard symmetry-preserving optimization assumptions.
Lahoti \emph{et al.}~\cite{lahotiRoleLocalityWeight2023} further separate the roles of locality and weight sharing (CNNs vs. locally connected neural networks (LCNs) vs. fully connected neural networks (FCNs)), showing that these architectural priors can produce clear sample-complexity gaps on translation-related tasks.
The operational consequence is direct: exposing the effect of inductive bias requires a reduced-data regime in which this bias is visible.
Accordingly, the main text uses an explicit translated-MNIST benchmark (1000 training examples per class, with digits resized to $16\times16$, placed on a $32\times32$ canvas, and translated by at most $8$ pixels per axis; Sec.~\ref{sec:mnist_motivation}) in addition to the full-MNIST size sweep.

\subsection{Classical convolution and translation symmetry}\label{sec:supp_classical_convolution}

A classical feedforward network is a composition of layers of the form $z = f(Ax)$, where $A$ is a linear operator and $f$ is a non-linear activation.
In many spatial problems, the same local pattern is applied at every position: each output feature at location $k$ depends on inputs in the same relative way at any location.

In one dimension this leads to the familiar convolutional form
\[
 y_k = \sum_{n=0}^{N-1} a_{k-n}\, x_n,
\]
where indices are taken modulo $N$.
Define the circular shift operators $T_k$ by
\[
(T_k x)_j = x_{j-k}, \qquad k\in\mathbb{Z},
\]
with indices modulo $N$.
A matrix $A$ is circulant if and only if it commutes with all shifts,
\[
T_k A = A T_k \qquad \forall k.
\]
In that case, multiplication by $A$ is exactly a circular convolution, $y = a \star x$.
Enforcing this commutation relation at the architectural level is what weight sharing means in classical CNNs, and it reduces the number of degrees of freedom from $N^2$ (dense) to $N$ (circulant), producing the classical convolutional inductive bias.

In practical models one works with multiple channels.
Then each scalar coefficient becomes a linear map on the channel space and the input/output acquire a channel index:
\[
(A z)_{k,j} = \sum_{n=0}^{N-1}\sum_l a_{k-n,j,l}\, z_{n,l}.
\]
For images the spatial grid is two- (or three-) dimensional.
In two dimensions, with translation symmetry along each axis, the shifts act separately on each coordinate.
For an input $z_{n_1,n_2}$ we write
\[
(T_k \otimes I)\, z_{n_1,n_2} = z_{n_1-k,n_2}, \qquad
(I \otimes T_k)\, z_{n_1,n_2} = z_{n_1,n_2-k}.
\]
A two-dimensional convolution layer commutes with all translations $T_x\otimes T_y$ and has the standard form
\begin{equation}\label{eq:convolution}
(A z)_{x,y,j} =
\sum_{n=0}^{N-1}\sum_{m=0}^{M-1}\sum_l
a_{x-n,y-m,j,l}\, z_{n,m,l}.
\end{equation}

A key structural fact is that circulant (and block-circulant) operators are diagonal in the Fourier basis.
Let $F_N$ be the discrete Fourier transform,
\begin{equation}\label{eq:classic_fourier}
(F_N x)_k =
\frac{1}{\sqrt{N}}\sum_{n=0}^{N-1} x_n \omega^{kn},
\qquad \omega = e^{-2\pi i / N}.
\end{equation}
In two dimensions, $(F_N\otimes F_M)^\dagger A (F_N\otimes F_M)$ becomes block-diagonal, with blocks acting on channels:
\begin{equation}\label{eq:circulant_classic}
(\hat A \hat z)_{p_1,p_2,j}
=
\sum_l \hat a_{p_1,p_2,j,l}\, \hat z_{p_1,p_2,l}.
\end{equation}
This Fourier characterization is the classical template we will mirror in the quantum construction.

\section{Proofs of the QCS--PCS mismatch and Fourier-multiplexer theorem}\label{sec:supp_symmetry_proofs}

This section gives complete proofs of the two structural statements used in the main text:
Lemma~\ref{lem:mismatch}, which separates qubit cyclic shifts (QCS) from pixel cyclic shifts (PCS) under address encoding, and
Theorem~\ref{thm:pcs_fourier}, which identifies the unitary commutant of PCS with a Fourier-mode multiplexer.
We use the notation of text Sec.~\ref{chapter:arch}.

\begin{lemma}[Detailed form of the QCS--PCS mismatch]\label{lem:supp_qcs_pcs_mismatch}
For pixel-to-qubit encodings, a cyclic translation of pixels can be represented by a cyclic permutation of physical qubits, up to the convention for shift direction.
For address/amplitude encodings with a nontrivial index register, the pixel cyclic shift is the modular-addition operator
\[
T\ket{j}=\ket{j+1 \bmod N},
\qquad N=2^n,
\]
on the computational basis of the index register.
This operator is not a cyclic permutation of the physical qubits of that register.
Consequently, QCS-equivariance does not in general imply PCS-equivariance.
\end{lemma}

\begin{proof}
In a pixel-to-qubit encoding, the computational tensor factors are indexed by pixels.
If a basis state is written as
\[
\ket{q_0}\ket{q_1}\cdots\ket{q_{N-1}},
\]
then a one-pixel cyclic translation sends the value stored at each pixel to the neighboring pixel.
As an operator on the encoded qubits this is precisely a cyclic permutation of tensor factors, either \(S\) or \(S^{-1}\) depending on whether the spatial shift is taken to the left or to the right.
Thus, in this encoding family, QCS can coincide with the pixel translation symmetry.

For address/amplitude encodings, the spatial coordinate is not a tensor-factor label but the binary value of the index register.
The pixel shift is therefore modular addition on the computational basis:
\[
T\ket{j}=\ket{j+1 \bmod N}.
\]
Every permutation of physical qubits leaves the all-zero basis state fixed:
\[
\Pi\ket{0}^{\otimes n}=\ket{0}^{\otimes n}
\]
for any qubit permutation \(\Pi\).
For \(N>1\), however,
\[
T\ket{0}^{\otimes n}=\ket{1}\neq \ket{0}^{\otimes n}.
\]
Here \(\ket{1}\) denotes the integer-one address state on the \(n\)-qubit index register.
Hence the address-encoded PCS operator \(T\) is not equal to any qubit permutation, and in particular it is not equal to a cyclic qubit shift.

It remains to show that commutation with a qubit cyclic shift need not imply commutation with \(T\).
For \(n\ge2\), take \(U=S\), where \(S\) is the cyclic permutation of the \(n\) index qubits.
Then \(U\) commutes with \(S\) trivially, but \(S\) and \(T\) do not commute:
\[
S T\ket{0}^{\otimes n}=S\ket{1}\neq \ket{1}=T S\ket{0}^{\otimes n}.
\]
For the one-qubit case, \(S=I\) while \(T=X\), so any unitary such as \(Z\) commutes with \(S\) but not with \(T\).
Thus QCS-equivariance alone does not enforce PCS-equivariance under address encoding.
\end{proof}

\begin{theorem}[Detailed form of the Fourier-multiplexer characterization]\label{thm:supp_pcs_fourier}
Let \(T\) be the cyclic shift on an index register \(\mathcal{H}_{\mathrm{idx}}\cong\mathbb{C}^N\), and let \(\mathcal{H}_{\mathrm{feat}}\) be an arbitrary finite-dimensional feature space.
A unitary \(U\) on \(\mathcal{H}_{\mathrm{idx}}\otimes\mathcal{H}_{\mathrm{feat}}\) commutes with \(T\otimes I\) if and only if
\[
U=(F_N^\dagger\otimes I)\left(\bigoplus_{k=0}^{N-1}U_k\right)(F_N\otimes I),
\]
where \(F_N\) is the \(N\)-point Fourier transform on the index register and each \(U_k\) is a unitary on \(\mathcal{H}_{\mathrm{feat}}\).
In two and three spatial dimensions, the same statement holds with tensor-product Fourier transforms and with blocks indexed by the corresponding tuple of Fourier modes.
\end{theorem}

\begin{proof}
We first prove the one-dimensional statement.
Let
\[
F_N\ket{j}=\frac{1}{\sqrt{N}}\sum_{k=0}^{N-1}\omega_N^{kj}\ket{k},
\qquad
\omega_N=e^{-2\pi i/N}.
\]
With the shift convention \(T\ket{j}=\ket{j+1\bmod N}\), direct substitution gives
\[
F_N T F_N^\dagger
=
\Omega
\coloneqq
\sum_{k=0}^{N-1}\omega_N^k\ket{k}\bra{k}.
\]
Define
\[
B\coloneqq (F_N\otimes I)U(F_N^\dagger\otimes I).
\]
Then
\[
[U,T\otimes I]=0
\quad\Longleftrightarrow\quad
[B,\Omega\otimes I]=0.
\]
The eigenvalues \(\omega_N^k\) of \(\Omega\) are all distinct.
Therefore each Fourier-mode projector \(P_k=\ket{k}\bra{k}\) is a polynomial in \(\Omega\), explicitly
\[
P_k
=
\prod_{r\ne k}
\frac{\Omega-\omega_N^r I}{\omega_N^k-\omega_N^r}.
\]
If \(B\) commutes with \(\Omega\otimes I\), it commutes with every \(P_k\otimes I\).
Hence \(B\) preserves each subspace \(\ket{k}\otimes\mathcal{H}_{\mathrm{feat}}\), and its off-diagonal Fourier-mode blocks vanish:
\[
(P_k\otimes I)B(P_l\otimes I)=0
\qquad (k\ne l).
\]
Thus
\[
B=\sum_{k=0}^{N-1}(P_k\otimes I)B(P_k\otimes I)
=
\bigoplus_{k=0}^{N-1}U_k,
\]
where
\[
U_k=(\bra{k}\otimes I)\,B\,(\ket{k}\otimes I)
\]
acts on \(\mathcal{H}_{\mathrm{feat}}\).
Because \(B\) is unitary and block diagonal, each block \(U_k\) is unitary.
Conjugating back by \(F_N\otimes I\) gives the claimed Fourier-multiplexer form.

Conversely, if
\[
B=\bigoplus_{k=0}^{N-1}U_k
\]
in the Fourier basis, then \(B\) is block diagonal with respect to the eigenspaces of \(\Omega\otimes I\).
It therefore commutes with \(\Omega\otimes I\), and the conjugated unitary
\[
U=(F_N^\dagger\otimes I)B(F_N\otimes I)
\]
commutes with \(T\otimes I\).

The multidimensional case is the simultaneous version of the same argument.
In two dimensions, set \(F=F_{N_x}\otimes F_{N_y}\) and
\[
B=(F\otimes I)U(F^\dagger\otimes I).
\]
Commutation with \(T_x\otimes I\) and \(T_y\otimes I\) is equivalent to commutation of \(B\) with
\[
\Omega_x\otimes I
\quad\text{and}\quad
\Omega_y\otimes I,
\]
where
\[
\Omega_x=\sum_{k_x,k_y}\omega_{N_x}^{k_x}
\ket{k_x,k_y}\bra{k_x,k_y},
\qquad
\Omega_y=\sum_{k_x,k_y}\omega_{N_y}^{k_y}
\ket{k_x,k_y}\bra{k_x,k_y}.
\]
Here \(\omega_M=e^{-2\pi i/M}\).
The joint spectral projector onto a fixed pair \((k_x,k_y)\) is the product of the one-dimensional spectral projectors for \(\Omega_x\) and \(\Omega_y\).
Therefore \(B\) commutes with every joint projector if and only if it is block diagonal over the joint Fourier modes:
\[
B=\bigoplus_{k_x=0}^{N_x-1}\bigoplus_{k_y=0}^{N_y-1}U_{k_x,k_y}.
\]
Unitarity of \(B\) again implies unitarity of all feature-space blocks.
The converse is immediate because such a block-diagonal \(B\) commutes with both diagonal translation operators.
The three-dimensional statement follows by adding the third commuting generator and the corresponding joint spectral projectors.
\end{proof}

The proof shows a slightly stronger algebraic fact: the same block-diagonal characterization holds for arbitrary linear maps in the commutant of the translation action.
Restricting to unitary maps simply restricts the Fourier-mode blocks to be unitary.

\section{MERA-QCNN templates and QCS equivariance}\label{sec:supp_mera_templates}

The main text (Sec.~\ref{sec:mera_qcnn}) uses MERA-QCNN as a representative example of a common QCNN design pattern that enforces cyclic-shift symmetry on physical qubits (QCS).
The longer discussion is included here.

\subsection{MERA-QCNN as a QCS-equivariant architecture}\label{sec:supp_mera_qcnn}

Figure~\ref{fig:general_conv}(a) shows a MERA-inspired QCNN, originally proposed in Ref.~\cite{Cong2019} and used in several subsequent works.
The architecture consists of repeated local unitaries arranged in a multiscale pattern, interleaved with pooling that reduces the number of active qubits, typically implemented by measurements followed by classically controlled operations (or, equivalently, by postponing measurement and replacing it by controlled gates due to standard circuit identities \cite{Nielsen_Chuang_2010}).

The essential structural ingredient is weight sharing along a line of qubits, but full one-site QCS equivariance also requires the layout itself to respect the cyclic action.
Many MERA/QCNN templates can be made QCS-equivariant by imposing cyclic weight sharing together with periodic boundary conditions, a periodic or symmetrized placement pattern containing the full cyclic orbit of each local block, and pooling/readout rules compatible with the same qubit permutation.
Without these additional choices, a standard brickwork layout generally commutes only with the subgroup preserved by its tiling and pooling pattern, rather than with the full cyclic shift $S$.
In our terminology, circuits satisfying this commutation property are QCS-QCNNs.
As Lemma~\ref{lem:mismatch} emphasizes, this is well matched to pixel-to-qubit encodings, but it does not in general enforce the translation action $T$ induced by address encoding, which is the setting of interest here.

\section{Multiplexer and Fourier-junction realization}\label{sec:supp_fourier_junction}

Two circuit-level points are useful for the implementation discussion: (i) how a block-diagonal multiplexer $\mathcal{B}$ can be decomposed into controlled operations, and (ii) how the inverse Fourier transform at the end of one layer and the forward Fourier transform at the start of the next layer collapse across pooling boundaries.

\begin{figure}
    \centering
    \includegraphics[width=0.5\textwidth]{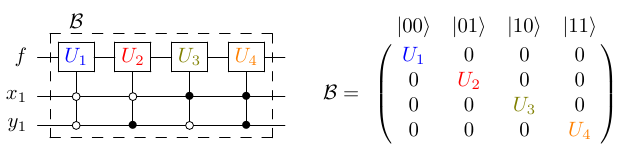}
    \caption{Decomposition of a block-diagonal multiplexer $\mathcal{B}$ into controlled operations $U_i$ on the feature qubits. The example uses a $2\times 2$ image (one qubit in each of the $x$ and $y$ registers) for clarity.}
    \label{fig:B_decomposition}
\end{figure}

\subsection{Fourier cancellation at the interface of PCS-QCNN layers}
\label{sec:supp_fourier_interface}

Figure~\ref{fig:2_layers} shows that consecutive PCS layers contain adjacent inverse and forward Fourier transforms on the
index registers.
In our architecture the pooling step is implemented after the inverse QFT of a layer: in the computational index basis
we split the fine index as $r=2q+s$, measure the parity bit $s$, discard it, and use the outcome to classically
condition the parameters of the next layer.
After previous pooling steps, the active index is already the corresponding coarse coordinate; the next pooling step again
measures the least significant bit of this current active computational index.
Equivalently, successive pooling steps remove the original computational-index bits from least to most significant order.
The next-layer QFT $F_{N/2}$ acts only on the surviving coarse index $q$.
In the reduced Fourier-domain representation this spatial-parity pooling appears as an aliasing operation on Fourier
modes: the same coarse Fourier label $q$ receives contributions from the pair $k=q$ and $k=q+N/2$.
As a result, the only nontrivial Fourier junction between two consecutive PCS layers reduces to a simple isometry on the
remaining index qubits, which can be implemented without explicit forward or inverse QFT blocks.

\begin{lemma}[Reduction of the Fourier junction]\label{lem:fourier_junction_reduction}
Let $N=2^n$ and let the pooling projection in the computational index basis be
\[
P_b \;:=\; \sum_{q=0}^{N/2-1}\ket{q}\bra{2q+b}
\;:\; \mathbb{C}^{N}\to\mathbb{C}^{N/2},
\]
where $b\in\{0,1\}$ is the measured parity bit in the split $r=2q+b$.
On the Fourier-input side, introduce the alias-ordering isomorphism
\[
A_N:\mathbb{C}^{N/2}\otimes\mathbb{C}^{2}\to\mathbb{C}^{N},
\qquad
A_N\ket{q}\ket{\sigma}=\ket{q+\sigma N/2},
\]
so that $\sigma$ distinguishes the two Fourier modes that alias to the same coarse label $q$.
For $b\in\{0,1\}$ define the postselected junction map
\begin{equation}
\label{eq:Kb_def}
K_b \;:=\; F_{N/2}\,P_b\,F_N^\dagger \;:\; \mathbb{C}^{N}\to \mathbb{C}^{N/2}.
\end{equation}
Then
\begin{equation}
\label{eq:Kb_closed}
K_b A_N \;=\; G^{\,b}\,(I\otimes \bra{b}H),
\end{equation}
where $H$ is the Hadamard acting on the Fourier alias selector $\sigma$, $G^{\,b}$ means $I$ for $b=0$ and $G$ for $b=1$, and $G$ is the diagonal phase-gradient operator on the remaining $(n-1)$-qubit register defined by
\begin{equation}
\label{eq:G_def}
G\ket{q} \;=\; e^{+2\pi i\, q/N}\ket{q}.
\end{equation}
Equivalently, for $\ket{\psi}=\sum_{k=0}^{N-1}\alpha_k\ket{k}$ one has
\begin{equation}
\label{eq:Kb_action}
K_b\ket{\psi}
=
\frac{1}{\sqrt{2}}
\sum_{q=0}^{N/2-1}
e^{+2\pi i\, q b/N}\big(\alpha_q+(-1)^b\alpha_{q+N/2}\big)\ket{q}.
\end{equation}
\end{lemma}

\begin{proof}
Let $\omega=e^{-2\pi i/N}$, so that
\[
F_N^\dagger\ket{k}
=
\frac{1}{\sqrt{N}}\sum_{r=0}^{N-1}\omega^{-kr}\ket{r}.
\]
With the decomposition $r=2q+s$ (so $\ket{r}\equiv\ket{q}\ket{s}$) we have
\[
P_bF_N^\dagger\ket{k}
=
\frac{1}{\sqrt{N}}\sum_{q=0}^{N/2-1}\omega^{-k(2q+b)}\ket{q}
=
\frac{\omega^{-kb}}{\sqrt{N}}\sum_{q=0}^{N/2-1}\omega^{-2kq}\ket{q}.
\]
Applying $F_{N/2}$ and using $\omega_{N/2}=e^{-2\pi i/(N/2)}=\omega^2$ gives, for each output basis state $\ket{p}$,
\[
\bra{p}K_b\ket{k}
=
\frac{1}{\sqrt{N(N/2)}}\,\omega^{-kb}\sum_{q=0}^{N/2-1}\omega^{2(p-k)q}.
\]
The inner sum is a geometric series.
Since $\omega^{2(p-k)}$ is an $(N/2)$-th root of unity, it evaluates to
\[
\sum_{q=0}^{N/2-1}\omega^{2(p-k)q}
=
\frac{N}{2}\,\delta_{p,\;k \bmod (N/2)}.
\]
Hence
\[
\bra{p}K_b\ket{k}
=
\frac{1}{\sqrt{2}}\,\omega^{-kb}\,\delta_{p,\;k \bmod (N/2)}.
\]
Writing $k=q+\sigma\frac{N}{2}$ with $\sigma\in\{0,1\}$ and $q\in\{0,\dots,N/2-1\}$, we obtain
\[
K_b\ket{q}
=
\frac{1}{\sqrt{2}}\,\omega^{-qb}\ket{q},\quad
K_b\ket{q+N/2}
=
\frac{1}{\sqrt{2}}\,\omega^{-(q+N/2)b}\ket{q}
=
\frac{(-1)^b}{\sqrt{2}}\,\omega^{-qb}\ket{q}.
\]
By linearity, for $\ket{\psi}=\sum_k\alpha_k\ket{k}$ this yields \eqref{eq:Kb_action}, since $\omega^{-qb}=e^{+2\pi i qb/N}$.
Finally, \eqref{eq:Kb_closed} follows in the alias ordering $A_N\ket{q}\ket{\sigma}=\ket{q+\sigma N/2}$:
$(I\otimes\bra{b}H)$ produces the factor
$(\alpha_q+(-1)^b\alpha_{q+N/2})/\sqrt{2}$ on each $\ket{q}$, while $G^b$ contributes the phase
$e^{+2\pi i qb/N}$ from \eqref{eq:G_def}.
\end{proof}

\begin{lemma}[Circuit implementation of the reduced junction]\label{lem:fourier_junction_circuit}
Let $q=\sum_{j=0}^{n-2} q_j 2^j$ be the binary expansion of the $(n-1)$-qubit index.
Then the phase-gradient operator \eqref{eq:G_def} can be implemented (up to a global phase) as a tensor product of single-qubit $Z$-rotations:
\begin{equation}
\label{eq:G_factor}
G \;\equiv\; \bigotimes_{j=0}^{n-2} R_z^{(j)}\!
\left(+\frac{\pi}{2^{\,n-1-j}}\right),
\end{equation}
where $R_z^{(j)}(\cdot)$ acts on qubit $j$ of the $(n-1)$-qubit index register.
Consequently, the reduced junction is implemented in the Fourier alias ordering by:
(i) apply $H$ to the alias selector $\sigma$ in $k=q+\sigma N/2$, (ii) measure it obtaining the same branch label $b$, and
(iii) conditionally on $b=1$ apply the gates from \eqref{eq:G_factor} to the remaining $(n-1)$ index qubits
(do nothing if $b=0$), i.e., implement $G^b$ with $G^0=I$ and $G^1=G$.
\end{lemma}

\begin{proof}
For $N=2^n$ and $q=\sum_{j=0}^{n-2} q_j2^j$,
\[
e^{+2\pi i q/N}
=
\prod_{j=0}^{n-2} e^{+i\pi q_j/2^{\,n-1-j}}.
\]
The single-qubit gate $R_z(\theta)=\exp(-i\theta Z/2)$ is diagonal and acts as
$\ket{0}\mapsto e^{-i\theta/2}\ket{0}$, $\ket{1}\mapsto e^{+i\theta/2}\ket{1}$.
Therefore, on a computational basis state $\ket{q}=\bigotimes_{j=0}^{n-2}\ket{q_j}$,
\[
\left(\bigotimes_{j=0}^{n-2} R_z(\theta_j)\right)\ket{q}
=
\left(\bigotimes_{j=0}^{n-2} e^{-i\theta_j/2}\right)
\left(\prod_{j=0}^{n-2} e^{+i\theta_j q_j}\right)\ket{q}.
\]
Choosing $\theta_j=+\pi/2^{\,n-1-j}$ yields the $q$-dependent factor
$\prod_{j} e^{+i\pi q_j/2^{\,n-1-j}}=e^{+2\pi i q/N}$.
The remaining prefactor is independent of $q$, hence it is a global phase.
\end{proof}

The same construction extends directly to two dimensions.
In the $2$D architecture (registers $\mathrm{Reg}\,x$ and $\mathrm{Reg}\,y$), spatial-parity pooling on the two
registers produces outcomes $(b_x,b_y)$, and the reduced Fourier junction factorizes as
$K_{b_x}^{(x)}\otimes K_{b_y}^{(y)}$ with $N\mapsto N_x$ and $N\mapsto N_y$ in
Lemmas~\ref{lem:fourier_junction_reduction}-\ref{lem:fourier_junction_circuit}.
Thus the interface of two consecutive layers can be implemented by two Hadamards, two measurements, and two independent
conditional phase-gradient blocks on the remaining index qubits, all in the corresponding Fourier alias orderings.

\begin{figure}[t]
\centering
\makebox[\columnwidth][c]{%
\hfill
\begin{minipage}[c]{0.22\columnwidth}
\centering
\includegraphics[width=\linewidth]{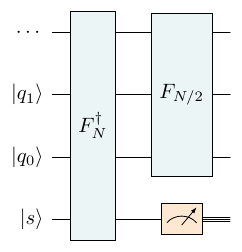}
\end{minipage}%
\hspace{10mm}%
\begin{minipage}[c]{0.03\columnwidth}
\centering
\Large{$=$}
\end{minipage}\hspace{10mm}%
\begin{minipage}[c]{0.22\columnwidth}
\centering
\includegraphics[width=\linewidth]{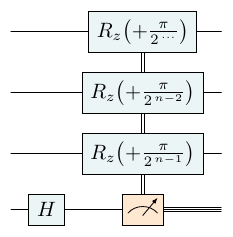}
\end{minipage}
\hfill}
\caption{Reduction of the Fourier junction between two adjacent PCS-QCNN layers.
Left: the original junction $F_{N/2}\,\mathcal{M}_s\,F_N^{\dagger}$, where $\mathcal{M}_s$ measures the spatial parity bit in the split $r=2q+s$.
Right: the reduced implementation from Lemma~\ref{lem:fourier_junction_circuit}: in the Fourier alias split $k=q+\sigma N/2$, apply $H$ and measure $\sigma$,
then conditionally (for outcome $b=1$) apply the single-qubit $Z$-rotations on the $q$-register.}
\label{fig:fourier_junction_reduction}
\end{figure}

The hybrid PCS-QCNN has the following canonical form after Fourier cancellation, in notation consistent with text Sec.~\ref{subsec:pcs_bp_model}.
Lemmas~\ref{lem:fourier_junction_reduction}-\ref{lem:fourier_junction_circuit} imply that the inverse Fourier transform at the end of one layer and the forward Fourier transform at the start of the next layer
can be replaced by a fixed junction map, so trainable parameters remain only inside
layer multiplexer blocks.

\smallskip
We describe the $2$D specialization ($d=2$) with registers $\mathrm{Reg}\,x$ and $\mathrm{Reg}\,y$.
Let $n_{\mathrm{idx}}$ be the number of index qubits per axis before the first layer, and let $Q$ be the number of quantum
layers.
At layer $\ell\in\{1,\dots,Q\}$, the number of active index qubits per axis is
$n_\ell=n_{\mathrm{idx}}-\ell+1$.
For $\ell<Q$, pooling measures one qubit per axis and produces
$m_{\ell}=(b_x^{(\ell)},b_y^{(\ell)})\in\{0,1\}^2$.

\smallskip
After cancellation of intermediate Fourier transforms, the only parameterized operations are the multiplexer blocks
\[
\mathcal{B}^{(\ell)}(m_{\ell-1})
=
\bigoplus_{k_x=0}^{2^{n_\ell}-1}
\bigoplus_{k_y=0}^{2^{n_\ell}-1}
V^{(\ell)}_{k_x,k_y}(m_{\ell-1}),
\]
with $m_0=0$ and $V^{(\ell)}_{k_x,k_y}(m_{\ell-1})\in U(D_f)$ acting on the feature register.
The inter-layer junction is fixed and parameter-free: it is implemented by local Hadamards, measurement of pooling qubits,
and conditional phase-gradient gates (Lemmas~\ref{lem:fourier_junction_reduction}-\ref{lem:fourier_junction_circuit}).

\smallskip
Equivalently, the same architecture can be written in layer form as
\[
U^{(\ell)}(m_{\ell-1})=
(\mathcal{F}^{(\ell)\dagger}\otimes I)\,
\mathcal{B}^{(\ell)}(m_{\ell-1})\,
(\mathcal{F}^{(\ell)}\otimes I),
\]
which is exactly Eq.~\eqref{eq:U_layer_def} in text Sec.~\ref{subsec:pcs_bp_model}.
The full measured-and-controlled process defines the output distribution
$p_{\Theta}(z\mid x)$ from text Eq.~\eqref{eq:p_def}.

\smallskip
The gradient-accounting discussion below uses the following fact:
all trainable parameters are inside $\mathcal{B}^{(\ell)}$, while inter-layer junction maps are fixed.
This separation is used in Sec.~\ref{sec:supp_trainability_full}.

\section{Idealized benchmark model and resource scaling}\label{sec:supp_resource_scaling}

The resource scaling summarized in text Sec.~\ref{sec:complexity} should be read together with the main modeling caveat: the reported PCS-QCNN is an idealized symmetry benchmark model.
It uses direct state initialization, dense-tensor statevector evolution, and explicit block-diagonal application of fully general Fourier multiplexers.
The following discussion gives qubit, depth, and shot-cost accounting for moving such a model toward hardware, where state preparation and multiplexer compilation become central costs.

\subsection{Complexity analysis}\label{sec:supp_complexity}

In the NISQ regime it is important to keep both qubit requirements and circuit depth in view.
We begin with the qubit count.
For an $N\times N$ image with $N=2^{n_{\mathrm{idx}}}$, address encoding uses two index registers with $n_{\mathrm{idx}}=\log_2 N$ qubits each.
In addition, we use a feature register of size $n_f$ qubits, with the grayscale/color qubit included in $n_f$ by convention.
Thus
\begin{equation}
n_{\mathrm{tot}}=2n_{\mathrm{idx}}+n_f = 2\log_2 N + n_f.
\end{equation}
In the reported $16\times16$ benchmark, $n_{\mathrm{idx}}=4$ and $n_f\in\{1,2,3\}$, hence $n_{\mathrm{tot}}\in\{9,10,11\}$.

The key point for circuit depth is that trainable and compiled operations are not the same thing.
After cancellation of intermediate Fourier transforms (Sec.~\ref{sec:supp_fourier_interface}), the trainable operations are precisely the multiplexer blocks $\mathcal{B}^{(\ell)}$.
However, on hardware one still has to implement the Fourier transforms and compile the multiplexers into one- and two-qubit gates.
The Fourier-transform contribution is modest at the reported image sizes.
For an $n$-qubit QFT, the standard circuit uses $n$ Hadamards and $n(n-1)/2$ controlled-phase gates (plus optional swaps), hence gate count $O(n^2)$ and depth $O(n^2)$ without additional parallelization~\cite{Nielsen_Chuang_2010,Coppersmith2002,QFT}.
In PCS-QCNN, QFT acts separately on the $x$ and $y$ registers, so at $n_x$ and $n_y$ qubits per axis this yields a per-layer overhead on the order of $O(n_x^2+n_y^2)$ two-qubit entangling gates.
For our $16\times16$ case ($n_x=n_y=4$), this cost is small compared to the multiplexer compilation below.

The multiplexer synthesis is the main depth bottleneck.
In the Fourier basis, a PCS layer applies a mode-wise block-diagonal unitary of the form
\begin{equation}
\mathcal{B}^{(\ell)}(m)
=
\sum_{k_x,k_y} \ketbra{k_x,k_y}{k_x,k_y}\otimes U^{(\ell)}_{k_x,k_y}(m),
\end{equation}
where each block $U^{(\ell)}_{k_x,k_y}(m)$ acts on the $n_f$ feature qubits.
This is a multi-target quantum multiplexer: a uniformly controlled $U(2^{n_f})$ gate with $n_c=n_x+n_y$ control qubits.
General decompositions of uniformly controlled gates scale exponentially in $n_c$ in the worst case: one needs $O(2^{n_c})$ controlled blocks, and if each $U^{(\ell)}_{k_x,k_y}(m)$ is treated as a generic $U(2^{n_f})$ unitary, the total two-qubit gate count is $O(2^{n_c}\,4^{n_f})$ up to architecture-dependent constants~\cite{Bergholm_2005,Shende_multiplexer}.
In our benchmark, $n_c=8$ and $n_f\le 3$, so this worst-case compilation cost is still moderate, but it grows quickly with image resolution.

This hardware compilation cost is not paid in the statevector experiments.
In our classical simulations we apply each $\mathcal{B}^{(\ell)}$ directly as a block-diagonal unitary map.
This benchmark evaluates the PCS-aligned inductive bias under within-family controls: the PCS-QCNN and RBC-QCNN have identical parameter counts, while the classical CNN and MLP controls are closely parameter-matched.
Hardware-oriented variants would likely restrict $U^{(\ell)}_{k_x,k_y}(m)$ (e.g., shallow local ansatzes on the feature register and/or parameter sharing across modes) to reduce depth while preserving exact PCS equivariance at the unitary-block level.

Finally, finite-shot readout introduces a separate resource cost.
The measurement cost is proportional to the shot budget $N_{\mathrm{shot}}$ required to estimate the readout distribution $p_{\Theta}(\cdot\mid x)$ with the desired accuracy.
This dependence is intrinsic: in hybrid quantum-classical learning, measurement precision is a computational resource and must be treated as part of the model design.

\section{Gradient-accounting derivation}
\label{sec:supp_trainability_full}

The elementary accounting estimate used in Sec.~\ref{sec:pcs_bp_theory} is derived here.
The estimate concerns the mean per-coordinate second moment.
The key point is that, in the depth-scaling family, the aggregate gradient-energy bound grows at most linearly in \(Q\), whereas the fully general Fourier multiplexer has exponentially many scalar coordinates.

\subsection{PCS-QCNN family and parameter count}
\label{sec:supp_gradient_scaling}

We work with an index register consisting of $d$ spatial axes.
Each axis initially has $n_{\mathrm{idx}}$ qubits.
A depth-$Q$ PCS-QCNN pools after the first $Q-1$ layers, each time measuring and discarding one computational-index qubit per axis.
After pooling, the remaining index register has
\begin{equation}
  n_l\coloneqq n_{\mathrm{idx}}-Q+1
\end{equation}
qubits per axis, hence
\begin{equation}
  D_{\mathrm{idx}}\coloneqq 2^{d n_l}.
\end{equation}
The feature register has $n_f$ qubits and dimension
\begin{equation}
  D_f\coloneqq 2^{n_f}.
\end{equation}
The final measurement is performed on the remaining index qubits and the feature register, so
\begin{equation}
  D_{\mathrm{out}}\coloneqq D_{\mathrm{idx}}D_f.
\end{equation}

In the depth-scaling family used in the main text,
\begin{equation}
  n_l,\ d,\ D_f,\ M\ \text{are fixed},
  \qquad
  n_{\mathrm{idx}}=n_l+Q-1.
  \label{eq:supp_depth_scaling_regime}
\end{equation}
Thus the readout dimension remains fixed as $Q$ increases.

At layer $\ell\in\{1,\dots,Q\}$, the active index width per axis is
\begin{equation}
  n_{\ell}=n_{\mathrm{idx}}-\ell+1=n_l+Q-\ell,
\end{equation}
and the active index dimension is
\begin{equation}
  N_{\ell}=2^{d n_{\ell}}.
\end{equation}
Let $\mathcal{F}^{(\ell)}$ denote the $d$-fold QFT on the active index register at layer $\ell$.
The branch label available to layer $\ell$ is drawn from
\begin{equation}
  \mathcal{M}_{\ell}\coloneqq
  \begin{cases}
    \{0\}, & \ell=1,\\
    \{0,1\}^{d}, & \ell\ge 2,
  \end{cases}
  \label{eq:branch_set}
\end{equation}
where the singleton label for $\ell=1$ represents the absence of a previous pooling outcome.
This is the benchmark convention: later multiplexers are conditioned on the most recently pooled $d$-bit outcome, not on the full measurement history.

The Fourier multiplexer at layer $\ell$ is
\begin{equation}
  \mathcal{B}^{(\ell)}(m_{\ell-1})
  \coloneqq
  \sum_{k\in[N_{\ell}]}\ket{k}\!\bra{k}\otimes V^{(\ell)}_k(m_{\ell-1}),
  \label{eq:B_def}
\end{equation}
where each block $V^{(\ell)}_k(m_{\ell-1})\in U(D_f)$ acts only on the feature register.
The PCS-equivariant unitary block is
\begin{equation}
  U^{(\ell)}(m_{\ell-1})
  \coloneqq
  \big(\mathcal{F}^{(\ell)\dagger}\otimes\one\big)
  \mathcal{B}^{(\ell)}(m_{\ell-1})
  \big(\mathcal{F}^{(\ell)}\otimes\one\big).
  \label{eq:supp_U_layer_def}
\end{equation}
For $\ell<Q$, this unitary is followed by parity pooling on the computational index register; the final layer is followed by measurement of the remaining index and feature registers.
The full measurement-and-feedforward process defines
\begin{equation}
  p_{\Theta}(z\mid x),\qquad z\in[D_{\mathrm{out}}],
  \label{eq:supp_p_def}
\end{equation}
where $\Theta$ denotes all quantum parameters.

The classical head is a linear map followed by a softmax,
\begin{equation}
  q(x)=\mathrm{softmax}\big(Wp_{\Theta}(\cdot\mid x)+b\big)\in\Delta^M.
\end{equation}
The unitary PCS blocks are exactly equivariant on their active index registers.
Pooling yields a multiscale/coarse equivariance structure rather than exact one-pixel equivariance of the complete readout/classifier pipeline.
This unconstrained linear head is not symmetry-tied across readout coordinates, so exact PCS equivariance is a property of the unitary quantum convolutional blocks rather than of the full classifier output.
In the reported classifier, translation invariance of the final labels is learned from data rather than imposed exactly.
For label $c\in\{1,\dots,M\}$ the loss is
\begin{equation}
  \mathcal{L}(\Theta,W,b\mid x,c)\coloneqq -\log q_c.
  \label{eq:supp_loss_def}
\end{equation}
The same notation applies to a minibatch or empirical-risk loss by replacing $\mathcal{L}$ with the corresponding average.

For a dataset or minibatch \(\mathcal{D}=\{(x_i,c_i)\}_{i=1}^{N}\), define the single-sample quantum gradients
\begin{equation}
  g_i\coloneqq \nabla_{\Theta}\mathcal{L}(\Theta,W,b\mid x_i,c_i).
\end{equation}
The gradient used by a first-order optimizer for the averaged loss is
\begin{equation}
  G_{\mathcal{D}}\coloneqq
  \nabla_{\Theta}\frac{1}{N}\sum_{i=1}^{N}\mathcal{L}(\Theta,W,b\mid x_i,c_i)
  =
  \frac{1}{N}\sum_{i=1}^{N}g_i .
  \label{eq:supp_empirical_gradient_def}
\end{equation}
By contrast, the per-sample RMS diagnostic is
\begin{equation}
  R_{\mathcal{D}}^2
  \coloneqq
  \frac{1}{N}\sum_{i=1}^{N}\|g_i\|_2^2 .
  \label{eq:supp_per_sample_rms_def}
\end{equation}
These are not interchangeable plateau diagnostics.
Taking expectation over random initialization,
\begin{equation}
  \mathbb{E}\|G_{\mathcal{D}}\|_2^2
  =
  \frac{1}{N^2}\sum_{i=1}^{N}\mathbb{E}\|g_i\|_2^2
  +
  \frac{1}{N^2}\sum_{i\ne j}\mathbb{E}\langle g_i,g_j\rangle .
  \label{eq:supp_empirical_gradient_cross_terms}
\end{equation}
Thus \(R_{\mathcal{D}}^2\) measures the diagonal, single-sample contribution, whereas \(\|G_{\mathcal{D}}\|_2^2\) also includes cross-sample alignment or cancellation.
For optimization of the empirical risk, \(\|G_{\mathcal{D}}\|_2\) is the directly relevant norm; \(R_{\mathcal{D}}\) is a diagnostic for whether individual samples still carry gradient signal.

For the full Pauli-exponential ansatz used in the benchmark, each feature block has
\begin{equation}
  p_{\mathrm{blk}}=4^{n_f}
\end{equation}
real scalar parameters.
The number of quantum-core parameters is
\begin{equation}
  P_Q=p_{\mathrm{blk}}\sum_{\ell=1}^{Q}|\mathcal{M}_{\ell}|N_{\ell}
  =p_{\mathrm{blk}}\sum_{\ell=1}^{Q} b_{\ell}\,2^{d(n_l+Q-\ell)},
  \qquad
  b_1=1,
  \quad b_{\ell}=2^d\ (\ell\ge 2).
  \label{eq:supp_PQ_count}
\end{equation}
In particular,
\begin{equation}
  P_Q\ge p_{\mathrm{blk}}\,2^{d(n_l+Q-1)}.
  \label{eq:supp_PQ_lower}
\end{equation}
For completeness, for $Q\ge2$ the exact closed form is
\begin{equation}
P_Q
=p_{\mathrm{blk}}\left[
2^{d(n_l+Q-1)}
+2^d\,2^{d n_l}\frac{2^{d(Q-1)}-1}{2^d-1}
\right],
\label{eq:supp_PQ_closed}
\end{equation}
so \(P_Q=\Theta(2^{dQ})\) at fixed \(n_l,d,n_f\).

\subsection{Aggregate gradient-energy estimate}
\label{subsec:supp_gradient_accounting}

The benchmark initialization samples every quantum Pauli coefficient independently as
\begin{equation}
  \theta^{(\ell)}_{k,m,\alpha}\sim\mathrm{Unif}(0,2\pi).
  \label{eq:supp_quantum_uniform_init}
\end{equation}
Independently, the classical head uses fan-in uniform initialization.
For head input dimension \(D_{\mathrm{out}}\), this means
\begin{equation}
  W_{az}\sim
  \mathrm{Unif}\!\left(-D_{\mathrm{out}}^{-1/2},D_{\mathrm{out}}^{-1/2}\right),
  \qquad
  b_a\sim
  \mathrm{Unif}\!\left(-D_{\mathrm{out}}^{-1/2},D_{\mathrm{out}}^{-1/2}\right),
  \qquad
  a\in[M],\ z\in[D_{\mathrm{out}}],
  \label{eq:supp_fanin_head_init}
\end{equation}
independently over all entries.
Because \(D_{\mathrm{out}}\) is fixed in Eq.~\eqref{eq:supp_depth_scaling_regime}, this head distribution is the same for every depth \(Q\).
The uniform quantum initialization fixes the experimental protocol, while the bounds below are deterministic in the quantum parameters and use only independence from the head.

\begin{lemma}[Softmax-head gradient bound]
\label{lem:supp_softmax_head_bound}
Let
\[
  q=\mathrm{softmax}(Wp+b)\in\Delta^M,
  \qquad
  \mathcal{L}(p,c)=-\log q_c .
\]
Then
\begin{equation}
  \|\nabla_p\mathcal{L}(p,c)\|_2^2
  \le
  2\|W\|_{\mathrm{op}}^2 .
  \label{eq:supp_softmax_head_bound}
\end{equation}
\end{lemma}

\begin{proof}
The standard cross-entropy/softmax derivative gives
\[
  \nabla_p\mathcal{L}=W^\top(q-e_c).
\]
Since $q\in\Delta^M$, one has $\|q-e_c\|_2^2\le2$, and the claim follows.
\end{proof}

\begin{lemma}[Fan-in uniform linear-softmax head: directional second-moment bound]
\label{lem:supp_fanin_head_directional_upper}
Assume the fan-in uniform head initialization~\eqref{eq:supp_fanin_head_init}, with \(L=D_{\mathrm{out}}\).
Fix \(p\in\Delta^{L}\), a label \(c\in[M]\), and a deterministic vector \(v\in\mathbb{R}^{L}\).
Then
\begin{equation}
  \mathbb{E}_{W,b}\!\left[
  \left\langle W^\top(\mathrm{softmax}(Wp+b)-e_c),v\right\rangle^2
  \mid p
  \right]
  \le
  \frac{2M}{3}\|v\|_2^2 .
  \label{eq:supp_fanin_head_directional_upper}
\end{equation}
\end{lemma}

\begin{proof}
Set \(q=\mathrm{softmax}(Wp+b)\) and \(r=q-e_c\).
For every realization of \(W,b\), \(\|r\|_2^2\le2\), and therefore
\[
  \left\langle W^\top r,v\right\rangle^2
  \le
  \|W^\top r\|_2^2\|v\|_2^2
  \le
  2\|W\|_{\mathrm{op}}^2\|v\|_2^2
  \le
  2\|W\|_F^2\|v\|_2^2 .
\]
Under~\eqref{eq:supp_fanin_head_init}, each weight has variance \(1/(3L)\), so
\[
  \mathbb{E}\|W\|_F^2
  =
  ML\,\frac{1}{3L}
  =
  \frac{M}{3}.
\]
Taking expectation gives~\eqref{eq:supp_fanin_head_directional_upper}.
\end{proof}

\begin{lemma}[Bounded tangents of Pauli-exponential blocks]
\label{lem:supp_pauli_exponential_tangent}
Let
\[
  U(\theta)=\exp\!\left(i\sum_{\alpha\in\mathcal{P}_{n_f}}\theta_{\alpha}P_{\alpha}\right),
\]
where \(\mathcal{P}_{n_f}\) is the full Pauli-string basis, including the all-identity string.
For each coordinate $\theta_{\alpha}$ there exists an operator $K_{\alpha}(\theta)$ such that
\begin{equation}
  \partial_{\theta_{\alpha}}U(\theta)=K_{\alpha}(\theta)U(\theta),
  \qquad
  \|K_{\alpha}(\theta)\|\le1.
  \label{eq:supp_pauli_tangent_bound}
\end{equation}
\end{lemma}

\begin{proof}
Write $A(\theta)=\sum_{\beta}\theta_{\beta}P_{\beta}$.
The derivative of the matrix exponential is
\[
  \partial_{\theta_{\alpha}}U(\theta)
  =
  i\int_0^1
  e^{i(1-s)A(\theta)}P_{\alpha}e^{isA(\theta)}\,ds .
\]
Each integrand has operator norm $\|P_{\alpha}\|=1$, hence $\|\partial_{\theta_{\alpha}}U(\theta)\|\le1$.
Taking $K_{\alpha}(\theta)=(\partial_{\theta_{\alpha}}U(\theta))U(\theta)^\dagger$ gives~\eqref{eq:supp_pauli_tangent_bound}.
\end{proof}

\begin{lemma}[Layer sensitivity-energy accounting]
\label{lem:supp_layer_sensitivity_energy}
Fix an input $x$, all model parameters, and a layer $\ell$.
For each coordinate $\theta^{(\ell)}_{k,\alpha}(m)$ in that layer, let
\[
  s^{(\ell)}_{k,m,\alpha}(x)
  \coloneqq
  \partial_{\theta^{(\ell)}_{k,\alpha}(m)}p_{\Theta}(\cdot\mid x)
  \in\mathbb{R}^{D_{\mathrm{out}}}.
\]
Then
\begin{equation}
  \sum_{m\in\mathcal{M}_{\ell}}
  \sum_{k\in[N_{\ell}]}
  \sum_{\alpha\in\mathcal{P}_{n_f}}
  \|s^{(\ell)}_{k,m,\alpha}(x)\|_2^2
  \le
  4p_{\mathrm{blk}}.
  \label{eq:supp_layer_sensitivity_energy}
\end{equation}
\end{lemma}

\begin{proof}
Use the deferred-measurement dilation of the pooling process.
For fixed $\ell$, write the final measured state as
\[
  \ket{\Psi(\theta)}=\mathcal{U}_{\Theta(\theta)}
  \big(\ket{\psi(x)}\otimes\ket{0}_{\mathrm{rec}}\big),
\]
and write the final readout probabilities as
\[
  p_{\Theta}(z\mid x)=\bra{\Psi}\Pi_z\ket{\Psi},
\]
where the orthogonal projectors $\{\Pi_z\}$ include the marginalization over deferred condition registers.
For one coordinate, let $\dot{\Psi}$ denote the derivative of $\ket{\Psi(\theta)}$ with respect to that coordinate.
Then
\[
  \partial_{\theta}p_{\Theta}(z\mid x)
  =
  2\Re\bra{\dot{\Psi}}\Pi_z\ket{\Psi},
\]
and orthogonality of the projectors gives
\[
  \|\partial_{\theta}p_{\Theta}(\cdot\mid x)\|_2^2
  \le 4\|\dot{\Psi}\|_2^2 .
\]

It remains to sum $\|\dot{\Psi}\|_2^2$ over the coordinates in layer $\ell$.
Immediately before the layer-$\ell$ multiplexer, let $\ket{\Phi}$ be the normalized dilated state, including the deferred record register.
For branch label $m$ and Fourier mode $k$, the derivative of the selected block has the form $K_{\alpha}V^{(\ell)}_k(m)$ with $\|K_{\alpha}\|\le1$ by Lemma~\ref{lem:supp_pauli_exponential_tangent}.
Thus, for an isometry $\mathcal{W}$ representing the rest of the circuit,
\[
  \dot{\Psi}
  =
  \mathcal{W}
  \big(\Pi_m\otimes\ket{k}\!\bra{k}\otimes K_{\alpha}V^{(\ell)}_k(m)\otimes I\big)
  \ket{\Phi},
\]
where $\Pi_m$ is the projector onto the relevant deferred pooling record (with the singleton branch at $\ell=1$).
Therefore
\[
  \|\dot{\Psi}\|_2^2
  \le
  \|(\Pi_m\otimes\ket{k}\!\bra{k}\otimes I)\ket{\Phi}\|_2^2 .
\]
For each fixed $\alpha$, the projectors over $(m,k)$ are mutually orthogonal and sum to at most the identity, so the sum of the right-hand side over $(m,k)$ is at most $1$.
Summing over the $p_{\mathrm{blk}}$ Pauli coordinates gives~\eqref{eq:supp_layer_sensitivity_energy}.
\end{proof}

To combine the two sensitivity estimates, consider first a single sample and condition on the quantum parameters, so the readout vector \(p=p_{\Theta}(\cdot\mid x)\) and all sensitivity vectors
\[
  s^{(\ell)}_{k,m,\alpha}(x)
  =
  \partial_{\theta^{(\ell)}_{k,\alpha}(m)}
  p_{\Theta}(\cdot\mid x)
\]
are fixed.
The chain rule gives
\[
  \partial_{\theta^{(\ell)}_{k,\alpha}(m)}\mathcal{L}
  =
  \left\langle
  W^\top(\mathrm{softmax}(Wp+b)-e_c),
  s^{(\ell)}_{k,m,\alpha}(x)
  \right\rangle .
\]
By Lemma~\ref{lem:supp_fanin_head_directional_upper},
\[
\mathbb{E}_{W,b}\!\left[
\left(\partial_{\theta^{(\ell)}_{k,\alpha}(m)}\mathcal{L}\right)^2
\mid \Theta
\right]
\le
\frac{2M}{3}
\|s^{(\ell)}_{k,m,\alpha}(x)\|_2^2 .
\]
Summing over all coordinates in one layer and applying Lemma~\ref{lem:supp_layer_sensitivity_energy} gives
\[
\mathbb{E}_{W,b}\!\left[
\sum_{m,k,\alpha}
\left(\partial_{\theta^{(\ell)}_{k,\alpha}(m)}\mathcal{L}\right)^2
\mid \Theta
\right]
\le
\frac{8M}{3}p_{\mathrm{blk}} .
\]
Summing over \(Q\) layers and then taking expectation over any independent quantum initialization gives
\begin{equation}
  \mathbb{E}\|\nabla_{\Theta_Q}\mathcal{L}(\Theta_Q,W,b)\|_2^2
  \le
  \frac{8M}{3}p_{\mathrm{blk}}Q .
  \label{eq:supp_aggregate_gradient_bound}
\end{equation}
For a minibatch or empirical-risk average, the quantum gradient is the corresponding average of the single-sample gradients, and convexity of the squared norm gives the same bound after averaging over samples.

Finally, let \(G_Q=\nabla_{\Theta_Q}\mathcal{L}(\Theta_Q,W,b)\).
Dividing~\eqref{eq:supp_aggregate_gradient_bound} by the deterministic parameter-count lower bound~\eqref{eq:supp_PQ_lower} yields the text accounting estimate
\begin{equation}
  \mathbb{E}\!\left[\frac{1}{P_Q}\|G_Q\|_2^2\right]
\le
  \frac{8M Q}{3\,2^{d(n_l+Q-1)}}.
\label{eq:supp_coordinate_counting_bound}
\end{equation}
For \(M=10\), the numerator constant is \(80Q/3\).
This is the parameter-count effect discussed in the main text.
The aggregate norm \(\mathbb{E}\|G_Q\|_2^2\) and the layerwise coordinates relevant to a particular optimizer or dataset are separate trainability quantities.
The main text therefore reports both the empirical-loss gradient norm and the per-sample RMS gradient norm directly.

}

\end{document}